\documentclass{ws-ijgmmp}{}

\begin{document}

\markboth{Sergiu I. Vacaru, El\c sen Veli Veliev and Enis Yaz\i c\i }
{A Geometric Method of Constructing Exact Solutions in Modified f(R,T) Gravity with YMH Interactions}
%
\catchline{}{}{}{}{}

\title
{A GEOMETRIC METHOD OF CONSTRUCTING EXACT SOLUTIONS IN MODIFIED f(R,T) GRAVITY WITH YANG--MILLS AND HIGGS INTERACTIONS}

\author{SERGIU I. VACARU}

\address{Theory Division, CERN, CH--1211, Geneva 23, Switzerland \footnote{visiting researcher}\\ and\\
Rector's Office, Alexandru Ioan Cuza University, Alexandru Lapu\c neanu stree, nr. 14,\\ UAIC - Corpus R, office 323;\ Ia\c si, Romania, 700057
 \\   {\qquad}\\
\email{sergiu.vacaru@cern.ch; sergiu.vacaru@uaic.ro}}

\author{EL\c SEN VELI VELIEV}

\address{Department of Physics, Kocaeli University, Izmit, 41380, Turkey
 \\   {\qquad}\\
\email{elsen@kocaeli.edu.tr}
}

\author{ENIS YAZICI}

\address{Department of Physics, Kocaeli University, Izmit, 41380, Turkey
 \\   {\qquad}\\
 \email{enis.yazici@kocaeli.edu.tr}
}

\maketitle

\begin{history}
\received{(Day Month Year)}
\revised{(Day Month Year)}
\end{history}

\begin{abstract}
We show that a geometric techniques can be elaborated and applied for constructing generic off-diagonal exact solutions in $f(R,T)$--modified
gravity for systems of gravitational--Yang-Mills--Higgs equations. The corresponding classes of metrics and generalized connections are determined
by generating and integration functions which depend, in general, on all space and time coordinates and may possess, or not, Killing symmetries. For
nonholonomic constraints resulting in Levi--Civita configurations, we can extract solutions of the Einstein--Yang--Mills--Higgs equations. We show
that the constructions simplify substantially for metrics with at least one Killing vector. There are provided and analyzed some examples of exact
solutions describing generic off--diagonal modifications to black hole/ellipsoid and solitonic configurations.

 \vskip0.1cm


{\footnotesize PACS 2008:\ 04.50.Kd, 04.90.+e, 11.15.Kc}

\end{abstract}






\section{Introduction}

The focus of this paper is the question if the fundamental gravitational
field and matter fields equations in theories of modified gravity can be
integrated in certain general forms with generic off--diagonal metrics and
(generalized) connections depending on all spacetime coordinates. \
Recently, a number of modified gravity theories have been elaborated with
the aim to explain the acceleration \cite{perl,riess} of universe and in
various attempts to formulate renormalizable quantum gravity models; for
reviews on $f(R)$--gravity and related theories, see \cite%
{odints1,odints2,odints3,vncbranes,covquant,vgrg}. The corresponding
classical field equations for the Einstein gravity and generalized/modified
gravity theories consist very sophisticate systems of nonlinear partial
differential equations (PDE).

The aim of this work is to present a treatment of modified gravitational
field and Yang--Mills--Higgs matter field equations, MGYMH, and\ developing
new geometric methods for integrating such PDE. Such solutions and methods
are not presented in well--known monographs on exact solutions in gravity
\cite{kramer,griff} and in Ref. \cite{halil}. We shall apply the so--called
anholonomic frame deformation method, AFDM, see \cite%
{ijgmmp,vexsol1,vexsol2,veymh} and references therein, for constructing
exact solutions of PDE in mathematical (modified) particle physics
describing generic off--diagonal and nonlinear gravitational, gauge and
scalar field interactions. A series of examples of off--diagonal solutions
for MGYMH black holes/ellpsoids and solitonic configurations will be
analyzed.

The main idea of the AFDM is to work with an auxiliary linear connection $%
\mathbf{D}$, which is necessary for decoupling the (modified) gravitational
field and matter field equations with respect to certain classes of adapted
frames with 2+2 splitting. This allows us to integrate physically important
PDE in very general off--diagonal forms. It is not possible to decouple such
PDE for the Levi--Civita connection, $\nabla ,$ and/or in arbitrary
coordinate frames. Nevertheless, we can extract exact and approximate
solutions for $\nabla $ by imposing nonholonomic constraints on $\mathbf{D}$
and corresponding generalized Ricci tensor, $\widehat{R}_{\alpha \beta },$
and scalar curvature, $\ ^{s}R.$ In this this approach, the $f(\ ^{s}R,T)$%
--modified gravity and related systems of gravitational--Yang-Mills--Higgs
(MGYMH) equations are determined by a functional $f(\ ^{s}R,T),$ where $T$ \
is the trace of the energy--momentum for matter fields. If $\mathbf{D}%
=\nabla $ and $f(\ ^{s}R,T)=R$ for the corresponding Levi--Civita
connection, $\nabla $, and corresponding curvature scalar, $R,$ we extract
solutions of the Einstein--Yang--Mills--Higgs (EYMH) equations. We shall
demonstrate that a number of physically important effects in modified
gravity theories can be equivalently modelled by off--diagonal interactions
(with certain nonholonomic constraints) in GR.

We shall proceed as follows. In Section \ref{s2} we provide geometric
preliminaries and prove the main Theorems on decoupling the modified
gravitational equations for generic off--diagonal metrics with one Killing
symmetry. There are considered also extensions to classes of "non--Killing"
solutions with coefficients depending on all set of four coordinates. We
show that the decoupling property also holds true for certain classes of
nonholonomic MGYMH systems. The theorems on generating off--diagonal
solutions for effectively polarized cosmological constants are considered in
Section \ref{s3}. Next two sections are devoted to examples of generic
off--diagonal exact solutions for MGYMH systems. In Section \ref{s4}, we
study the geometry of nonholonomic YMH vacuum deformations of black holes.
Certain examples of ellipsoid--solitonic non--Abbelian configurations and
related modifications are analyzed in Section \ref{s5}. Appendix \ref{sa}
contains some most important component formulas in the geometry of
nonholonomic manifolds 2+2 splitting. A proof of the decoupling property is
given in Appendix \ref{sb}.

\section{Nonholonomic Splitting in Modified Gravity}

\subsection{Geometric preliminaries}

\label{s2} We consider modified gravity theories, when the geometric models
of curved spacetime $\left( V,\mathbf{g,D}\right) $ are determined by a
generalized pseudo--Riemannian manifold $V$ endowed with a (modified)
Lorentzian metric $\mathbf{g}$ and a metric compatible linear connection $%
\mathbf{D,}$ when $\mathbf{Dg}=0.$ The geometric/physical data $(\mathbf{g,D)%
}$ are supposed to define solutions of certain systems of gravitational
field equations, see below section \ref{ssmgfe}, and $(\mathbf{g,D=\nabla )}$
as solutions of the Einstein equations.\footnote{%
We suppose that readers are familiar with basic concepts from differential
geometry and physical mathematics.}

\subsubsection{Metrics adapted to nonholonomic 2+2 splitting}

On a coordinate neighborhood $U\subset V,$ the local coordinates $%
u=\{u^{\alpha }=(x^{i},y^{a})\}$ are considered with conventional $2+2$
splitting into h-coordinates, $x=(x^{i}),$ and v-coordinates, $y=(y^{a}),$
for $j,k,...=1,2$ and $a,b,c,...=3,4.$ The local coordinate basis and
cobases are respectively $e_{\alpha }=\partial _{\alpha }=\partial /\partial
u^{\beta }$ and $e^{\beta }=du^{\beta },$ which can transformed into
arbitrary frames (tetrads/vierbeinds) via transforms of type $e_{\alpha
^{\prime }}=e_{\ \alpha ^{\prime }}^{\alpha }(u)e_{\alpha }$ and $e^{\alpha
^{\prime }}=e_{\alpha \ }^{\ \alpha ^{\prime }}(u)e^{\alpha }.$\footnote{%
The summation rule on repeating low--up indices will be applied if the
contrary will be not stated.} The coefficients of a vector $X$ and a metric $%
\mathbf{g}$ are defined, respectively, in the forms $X=X^{\alpha }e_{\alpha
} $ and
\begin{equation}
\mathbf{g}=g_{\alpha \beta }(u)e^{\alpha }\otimes e^{\beta },  \label{mst}
\end{equation}%
where $g_{\alpha \beta }:=\mathbf{g}(e_{\alpha },e_{\beta }).$ For our
purposes, we shall work with bases with non--integrable (equivalently,
nonholonomic/anholonomic) h--v--decomposition, when for the tangent bundle $%
TV$ $:=\bigcup\nolimits_{u}T_{u}V$ a Whitney sum
\begin{equation}
\mathbf{N}:\ TV=hV\oplus vV  \label{ncon}
\end{equation}%
is globally defined. In local form, this is defined by a nonholonomic
distribution with coefficients $N_{i}^{a}(u),$ when $\mathbf{N=}%
N_{i}^{a}(x,y)dx^{i}\otimes \partial /\partial y^{a}.$\footnote{%
For varios theories of gravity and different variables in corresponding
geometric models, a h--v--splitting defines a nonlinear connection,
N--connection, structures. On (pseudo) Riemannian manifolds to consider such
a N--connection is equivalent to a prescription of class of N--elongated
frames.}

It is possible to adapt the geometric constructions to a N--splitting (\ref%
{ncon}) if we work with \textquotedblright N--elongated\textquotedblright\
local bases (partial derivatives), $\mathbf{e}_{\nu }=(\mathbf{e}%
_{i},e_{a}), $ and cobases (differentials), $\mathbf{e}^{\mu }=(e^{i},%
\mathbf{e}^{a}),$ when
\begin{eqnarray}
\mathbf{e}_{i} &=&\partial /\partial x^{i}-\ N_{i}^{a}(u)\partial /\partial
y^{a},\ e_{a}=\partial _{a}=\partial /\partial y^{a},  \label{nader} \\
\mbox{ and  }e^{i} &=&dx^{i},\ \mathbf{e}^{a}=dy^{a}+\ N_{i}^{a}(u)dx^{i}.
\label{nadif}
\end{eqnarray}
For instance, the basic vectors (\ref{nader}) satisfy the nonholonomy
relations
\begin{equation}
\lbrack \mathbf{e}_{\alpha },\mathbf{e}_{\beta }]=\mathbf{e}_{\alpha }%
\mathbf{e}_{\beta }-\mathbf{e}_{\beta }\mathbf{e}_{\alpha }=W_{\alpha \beta
}^{\gamma }\mathbf{e}_{\gamma },  \label{nonholr}
\end{equation}%
with (antisymmetric) nontrivial anholonomy coefficients
\begin{equation}
W_{ia}^{b}=\partial _{a}N_{i}^{b},W_{ji}^{a}=\Omega _{ij}^{a}=\mathbf{e}%
_{j}\left( N_{i}^{a}\right) -\mathbf{e}_{i}(N_{j}^{a}).  \label{anhcoef}
\end{equation}%
The \textquotedblright boldface\textquotedblright\ letters are used in order
to emphasize that a spacetime model $\left( \mathbf{V},\mathbf{g,D}\right) $
and certain geometric objects and constructions are \textquotedblright
N--adapted\textquotedblright , i. e. adapted to a h--v--splitting. The
geometric objects are called distinguished (in brief, d--objects,
d--vectors, d--tensors etc) if they are adapted to the N--connection
structure via decompositions with respect to frames of type (\ref{nader})
and (\ref{nadif}). For instance, we write a d--vector as $\mathbf{X}=(hX,vX)$
and a d--metric as $\mathbf{g}=(hg,vg).$

\begin{proposition}
-\textbf{Definition.} Any spacetime metric $\mathbf{g}$ (\ref{mst}) can be
represented equivalently as a d--metric,
\begin{equation}
\ \mathbf{g}=\ g_{ij}(x,y)\ e^{i}\otimes e^{j}+\ g_{ab}(x,y)\ \mathbf{e}%
^{a}\otimes \mathbf{e}^{b},  \label{dm}
\end{equation}%
for $hg=\{\ g_{ij}\}$ and $\ vg=\{g_{ab}\}.$
\end{proposition}

\begin{proof}
Via frame transforms, $e_{\alpha }=e_{\ \alpha }^{\alpha ^{\prime
}}(x,y)e_{\alpha ^{\prime }},$ $\underline{g}_{\alpha \beta }=e_{\ \alpha
}^{\alpha ^{\prime }}e_{\ \beta }^{\beta ^{\prime }}\ \underline{g}_{\alpha
^{\prime }\beta ^{\prime }},$ any metric $\mathbf{g}=\underline{g}_{\alpha
\beta }du^{\alpha }\otimes du^{\beta }$ can be parameterized in the form
\begin{equation}
\underline{g}_{\alpha \beta }=\left[
\begin{array}{cc}
g_{ij}+N_{i}^{a}N_{j}^{b}g_{ab} & N_{j}^{e}g_{ae} \\
N_{i}^{e}g_{be} & g_{ab}%
\end{array}%
\right]  \label{ansatz}
\end{equation}%
for any prescribed set of coefficients $N_{i}^{a}.$ A metric $\mathbf{g}$ is
generic off--diagonal if the matrix (\ref{ansatz}) can not be diagonalized
via coordinate transforms. For such a metric, the anholonomy coefficients (%
\ref{anhcoef}) are not zero. Regrouping the terms for co--bases (\ref{nadif}%
), we obtain the d--metric (\ref{dm}). Parameterizations of this type are
used in Kaluza--Klein gravity when $N_{i}^{a}(x,y)=\Gamma _{bi}^{a}(x)y^{a}$
and $y^{a}$ are \textquotedblright compactified\textquotedblright\
extra--dimensions coordinates, or in Finsler gravity theories, see details
in \cite{vexsol2}. In this work, we restrict our considerations only to four
dimensional (4--d) gravity theories. $\square $
\end{proof}

\vskip5pt

We shall study exact solutions with metrics which via frame/ coordinate
transforms can be related to a d--metric $\mathbf{g}$ (\ref{dm}) and/or
ansatz (\ref{ansatz}) and written in a form with separation of
v--coordinates, $y^{3}$ and $y^{4},$ and nontrivial vertical conformal
transforms,
\begin{eqnarray}
\mathbf{g} &=&g_{i}dx^{i}\otimes dx^{i}+\omega ^{2}h_{a}\underline{h}_{a}%
\mathbf{e}^{a}\otimes \mathbf{e}^{a},  \label{ans1} \\
\mathbf{e}^{3} &=&dy^{3}+\left( w_{i}+\underline{w}_{i}\right) dx^{i},\
\mathbf{e}^{4}=dy^{4}+\left( n_{i}+\underline{n}_{i}\right) dx^{i},  \nonumber
\end{eqnarray}%
where%
\begin{eqnarray}
g_{i} &=&g_{i}(x^{k}),g_{a}=\omega ^{2}(x^{i},y^{c})\ h_{a}(x^{k},y^{3})%
\underline{h}_{a}(x^{k},y^{4}),  \nonumber  \\
N_{i}^{3} &=&w_{i}(x^{k},y^{3})+\underline{w}%
_{i}(x^{k},y^{4}),N_{i}^{4}=n_{i}(x^{k},y^{3})+\underline{n}%
_{i}(x^{k},y^{4}),  \label{paramdcoef}
\end{eqnarray}%
are functions of necessary smooth class which will be defined in a form to
generate solutions of certain fundamental gravitational and matter field
equations.\footnote{%
There is not summation on repeating \textquotedblright
low\textquotedblright\ indices $a$ in formulas (\ref{paramdcoef}) but such a
summation is considered for crossing \textquotedblright
up--low\textquotedblright\ indices $i$ and $a$ in (\ref{ans1})). We shall
underline a function if it positively depends on $y^{4}$ but not on $y^{3}$
and write, for instance, $\underline{n}_{i}(x^{k},y^{4})$.}

\subsubsection{N--adapted connections}

Linear connection structures can be introduced on a generalized spacetime $%
\mathbf{V}$ in form which is N--adapted to a h--v--splitting (\ref{ncon}),
or not.

\begin{definition}
A d--connection $\mathbf{D}=(hD,vD)$ is a linear connection preserving under
parallelism the N--connection structure.
\end{definition}

Any d--connection defines a covariant N--adapted derivative $\mathbf{D}_{%
\mathbf{X}}\mathbf{Y}$ of a d--vector field $\mathbf{Y}$ in the direction of
a d--vector $\mathbf{X}.$ With respect to N--adapted frames (\ref{nader})
and (\ref{nadif}), any $\mathbf{D}_{\mathbf{X}}\mathbf{Y}$ can be computed
as in GR but with the coefficients of the Levi--Civita connection
substituted by $\mathbf{D}=\{\mathbf{\Gamma }_{\ \alpha \beta }^{\gamma
}=(L_{jk}^{i},L_{bk}^{a},C_{jc}^{i},C_{bc}^{a})\}.$

Any d--connection is characterized by three fundamental geometric objects:
the d--torsion field which is (by definition)
\begin{equation*}
\mathcal{T}(\mathbf{X,Y}):=\mathbf{D}_{\mathbf{X}}\mathbf{Y}-\mathbf{D}_{%
\mathbf{Y}}\mathbf{X}-[\mathbf{X,Y}];
\end{equation*}%
the d--curvature field,
\begin{equation*}
\mathcal{R}(\mathbf{X,Y}):=\mathbf{D}_{\mathbf{X}}\mathbf{D}_{\mathbf{Y}}-%
\mathbf{D}_{\mathbf{Y}}\mathbf{D}_{\mathbf{X}}-\mathbf{D}_{\mathbf{[X,Y]}};
\end{equation*}%
and the nonmetricity field is $\mathcal{Q}(\mathbf{X}):=\mathbf{D}_{\mathbf{X%
}}\mathbf{g.}$ We compute the N--adapted coefficients of these geometric
objects by introducing $\mathbf{X}=\mathbf{e}_{\alpha }$ and $\mathbf{Y}=%
\mathbf{e}_{\beta },$ defined by (\ref{nader}), and $\mathbf{D}=\{\mathbf{%
\Gamma }_{\ \alpha \beta }^{\gamma }\}$ into above formulas,
\begin{eqnarray*}
\mathcal{T} &=&\{\mathbf{T}_{\ \alpha \beta }^{\gamma }=\left( T_{\
jk}^{i},T_{\ ja}^{i},T_{\ ji}^{a},T_{\ bi}^{a},T_{\ bc}^{a}\right) \}; \\
\mathcal{R} &\mathbf{=}&\mathbf{\{R}_{\ \beta \gamma \delta }^{\alpha }%
\mathbf{=}\left( R_{\ hjk}^{i}\mathbf{,}R_{\ bjk}^{a}\mathbf{,}R_{\ hja}^{i}%
\mathbf{,}R_{\ bja}^{c}\mathbf{,}R_{\ hba}^{i},R_{\ bea}^{c}\right) \mathbf{%
\};}\ \mathcal{Q}=\mathbf{\{Q}_{\ \alpha \beta }^{\gamma }\}.
\end{eqnarray*}%
The formulas for $\mathbf{T}_{\ \alpha \beta }^{\gamma }$ and $\mathbf{R}_{\
\beta \gamma \delta }^{\alpha }$ are similar, respectively, to (\ref{dtors})
and (\ref{dcurv}), see Appendix, but without "hats" on geometric objects.

\begin{theorem}
-\textbf{Definition. }There is a canonical d--connection $\widehat{\mathbf{D}%
}$ uniquely determined by any given $\mathbf{g}=\{g_{\alpha \beta }\}$ for a
prescribed $\mathbf{N}=\{N_{i}^{a}\}$ which is metric compatible, $\widehat{%
\mathbf{D}}\mathbf{g=0,}$ and with zero h--torsion, $h\widehat{\mathcal{T}}%
=\{\widehat{T}_{\ jk}^{i}\}=0,$ and zero v--torsion, $v\widehat{\mathcal{T}}%
=\{\widehat{T}_{\ bc}^{a}\}=0.$
\end{theorem}

\begin{proof}
It follows from a straightforward verification in N--adapted frames that $%
\widehat{\mathbf{D}}=\{$ $\widehat{\mathbf{\Gamma }}_{\ \alpha \beta
}^{\gamma }=(\widehat{L}_{jk}^{i},\widehat{L}_{bk}^{a},\widehat{C}_{jc}^{i},%
\widehat{C}_{bc}^{a})\}$ \ with coefficients
\begin{eqnarray}
\widehat{L}_{jk}^{i} &=&\frac{1}{2}g^{ir}\left( \mathbf{e}_{k}g_{jr}+\mathbf{%
e}_{j}g_{kr}-\mathbf{e}_{r}g_{jk}\right) ,  \label{cdc} \\
\widehat{L}_{bk}^{a} &=&e_{b}(N_{k}^{a})+\frac{1}{2}g^{ac}\left( \mathbf{e}%
_{k}g_{bc}-g_{dc}\ e_{b}N_{k}^{d}-g_{db}\ e_{c}N_{k}^{d}\right) ,  \nonumber  \\
\widehat{C}_{jc}^{i} &=&\frac{1}{2}g^{ik}e_{c}g_{jk},\ \widehat{C}_{bc}^{a}=%
\frac{1}{2}g^{ad}\left( e_{c}g_{bd}+e_{b}g_{cd}-e_{d}g_{bc}\right) ,  \nonumber
\end{eqnarray}
computed for a d--metric $\mathbf{g}=[g_{ij},g_{ab}]$ \ (\ref{dm}). $\square
$
\end{proof}

\vskip5pt

Every metric field $\mathbf{g}$ naturally and completely defines a
Levi--Civita connection $D=\nabla =\{\Gamma _{\ \alpha \beta }^{\gamma }\}$
if and only if there are satisfied the metric compatibility, $\ ^{\nabla }%
\mathcal{Q}(X)=\nabla _{\mathbf{X}}\mathbf{g}=0,$ and zero torsion, $\
^{\nabla }\mathcal{T}=0,$ conditions.\footnote{%
We emphasize that for geometric/physical objects defined by $\nabla $ we do
not use "boldface" symbols because this linear connection does not preserve
under parallelism and general frame/coordinate transforms a N--splitting (%
\ref{ncon}).} \ For $\widehat{\mathbf{D}}$, we can find a unique distortion
relation
\begin{equation}
\nabla =\widehat{\mathbf{D}}+\widehat{\mathbf{Z}},  \label{distrel}
\end{equation}%
where both linear connections $\nabla $ and $\widehat{\mathbf{D}}$ and the
distortion tensor $\widehat{\mathbf{Z}}$ are completely defined by $\mathbf{g%
}=\{g_{\alpha \beta }\}$ for a prescribed $\mathbf{N}=\{N_{i}^{a}\},$ see
details in \cite{ijgmmp,vexsol1,vexsol2} and Theorem \ref{thadist}. With
respect to N--adapted frames, the formula (\ref{distrel}) is given by (\ref%
{distrel1}) when the coefficients of the distortion tensor $\widehat{\mathbf{%
Z}}$ is determined by values (\ref{deft}). It is important to emphasize that
such formulas define a unique deformation of the Christoffel symbols for $%
\nabla $ into the corresponding coefficients (\ref{cdc}) of $\widehat{%
\mathbf{D}}$ because all such values are completely defined by the
coefficients of a metric tensor (\ref{ansatz}) (equivalently, (\ref{dm})).
In N--adapted form, the coefficients of the Levi--Civita connection (i.e.
the second type Christoffel symbols) can be computed by taking respective
sums of (\ref{cdc}) and (\ref{deft}).

The curvature, Ricci and Einstein tensors of the connections $\nabla $ and $%
\widehat{\mathbf{D}}$ are computed respectively by standard formulas (\ref%
{dcurv}), (\ref{driccic}) and (\ref{enstdt}) (see details in Appendix \ref%
{sa}). For instance, the curvature $\ ^{\nabla }\mathcal{R}$ for $\nabla $
can be computed as a distortion of $\widehat{\mathcal{R}}=\{\widehat{\mathbf{%
R}}_{\ \beta \gamma \delta }^{\alpha }\}$ (\ref{dcurv}) for $\widehat{%
\mathbf{D}}$ using the distortion relation for connections (\ref{distrel})
(and (\ref{distrel1})). The properties of the Ricci tensor of $\widehat{%
\mathbf{D}}$ are stated by Proposition \ref{approp}, see also formulas (\ref%
{riccie}) and (\ref{riccid}). Here we note that, in general, the Ricci
tensor $\widehat{\mathbf{R}}_{\ \beta \gamma }$ is not not symmetric because
of nontrivial nonholonomically induced torsion $\widehat{\mathcal{T}}.$
Nevertheless, such nonsymetric contributions do not result in nonsymmetric
metrics if we consider nonholonomic deformations determined by distortions (%
\ref{distrel}) computed by symmetric metrics when nonysmmetric components of
$\widehat{\mathbf{R}}_{\ \beta \gamma }$ are zero.

We conclude that all geometric constructions and physical theories derived
for the geometric data $\left( \mathbf{g,}\nabla \right) $ can be
equivalently modeled by the data $(\mathbf{g,N,}\widehat{\mathbf{D}})$
because of a unique distortion relation (\ref{distrel}). If we work with $%
\widehat{\mathbf{D}},$ we have non--trivial nonholonomically induced torsion
coefficients (\ref{dtors}). The meaning of such a torsion structure is
different from that, for instance, in Riemann--Cartan geometry when certain
additional spin like sources are considered for additional algebraic type
field equations for torsion fields. In our approach, $\widehat{\mathcal{T}}$
\ is completely defined by the metric structure when a N--splitting is
prescribed.

\subsection{Nonholonomic MG field equations}

\label{ssmgfe}We study modified gravity theories derived for the action%
\begin{equation}
S=\frac{1}{16\pi }\int \delta u\sqrt{|\mathbf{g}_{\alpha \beta }|}[f(\ ^{s}%
\widehat{R},T)+\ ^{m}L],  \label{act}
\end{equation}%
where $\ ^{m}L$ is the matter Lagrangian density; the stress--energy tensor
of matter is computed via variation on inverse metric tensor, $\ ^{m}\mathbf{%
T}_{\alpha \beta }=-\frac{2}{\sqrt{|\mathbf{g}_{\mu \nu }|}}\frac{\delta (%
\sqrt{|\mathbf{g}_{\mu \nu }|}\ ^{m}L)}{\delta \mathbf{g}^{\alpha \beta }}$,
trace $\ T:=\mathbf{g}^{\alpha \beta }\ ^{m}\mathbf{T}_{\alpha \beta },$ and
$f(\ ^{s}\widehat{R},\ ^{m}T)$ is an arbitrary functional on $\ ^{s}\widehat{%
R}$ (\ref{sdcurv}) and $\ ^{m}T.$ The volume form $\delta u$ is determined
by a d--metric $\mathbf{g}$ (\ref{dm}) in order to derive variational
formulas in N--adapted form.  For simplicity, we can assume that in
cosmological models the stress--energy tensor of the matter is given by
\begin{equation}
\ ^{m}\mathbf{T}_{\alpha \beta }=(\rho +p)\mathbf{v}_{\alpha }\mathbf{v}%
_{\beta }-p\mathbf{g}_{\alpha \beta },  \label{emt}
\end{equation}%
where in the approximation of perfect fluid matter $\rho $ is the energy
density, $p$ is the pressure and the four--velocity $\mathbf{v}_{\alpha }$
is subjected to the conditions $\mathbf{v}_{\alpha }\mathbf{v}^{\alpha }=1$
and $\mathbf{v}^{\alpha }\widehat{\mathbf{D}}_{\beta }\mathbf{v}_{\alpha }=0,
$ for $\ ^{m}L=-p$ in a corresponding local frame. We also consider
approximations of type%
\begin{equation}
f(\ ^{s}\widehat{R},\ ^{m}T)=\ ^{1}f(\ ^{s}\widehat{R})+\ ^{2}f(\ ^{m}T)
\label{functs}
\end{equation}%
and denote by $\ ^{1}F(\ ^{s}\widehat{R}):=\partial \ ^{1}f(\ ^{s}\widehat{R}%
)/\partial \ ^{s}\widehat{R}$ and $\ ^{2}F(\ ^{m}T):=\partial \ ^{2}f(\
^{m}T)/\partial \ \ ^{m}T.$

In this work, we consider effective sources parameterized with respect to
N--adapted frames in the form
\begin{equation}
\mathbf{\Upsilon }_{\beta \delta }=\ ^{ef}\eta \ G\ \ ^{m}\mathbf{T}_{\beta
\delta }+\ ^{ef}\mathbf{T}_{\beta \delta },  \label{effectsource}
\end{equation}%
with effective polarization of cosmological constant $\ ^{ef}\eta =[1+\
^{2}F/8\pi ]/\ ^{1}F$ and where the $f$--modification of the
energy--momentum tensor is computed as an additional effective source
\begin{equation}
\ ^{ef}\mathbf{T}_{\beta \delta }=[\frac{1}{2}(\ ^{1}f-\ ^{1}F\ ^{s}\widehat{%
R}+2p\ ^{2}F+\ ^{2}f)\mathbf{g}_{\beta \delta }-(\mathbf{g}_{\beta \delta }\
\widehat{\mathbf{D}}_{\alpha }\widehat{\mathbf{D}}^{\alpha }-\widehat{%
\mathbf{D}}_{\beta }\widehat{\mathbf{D}}_{\delta })\ ^{1}F]/\ ^{1}F.
\label{efm}
\end{equation}

We can postulate in geometric form or prove following a variational approach:

\begin{theorem}
The gravitational field equations for a modified gravity model (\ref{act})
with $f$--functional (\ref{functs}) and perfect fluid stress--energy tensor (%
\ref{emt}) can be re--written equivalently using the canonical d--connection
$\widehat{\mathbf{D}},$
\begin{equation}
\widehat{\mathbf{R}}_{\ \beta \delta }-\frac{1}{2}\mathbf{g}_{\beta \delta
}\ ^{s}R=\mathbf{\Upsilon }_{\beta \delta },  \label{cdeinst}
\end{equation}%
where the source d--tensor $\mathbf{\Upsilon }_{\beta \delta }$ is such way
constructed that $\mathbf{\Upsilon }_{\beta \delta }\rightarrow 8\pi
GT_{\beta \delta }$ \ for $\widehat{\mathbf{D}}\rightarrow \nabla ,$ where $%
T_{\beta \delta }$ is the energy--momentum tensor in GR with coupling
gravitational constant $G.$
\end{theorem}

\begin{proof}
By varying the action $S$ (\ref{act}) with respect to $\mathbf{g}^{\alpha
\beta }$ and following covariant differential calculus with N--elongated
operators (\ref{nader}) and (\ref{nadif}), we obtain the gravitational field
equations (\ref{cdeinst}) and respective effective sources (\ref%
{effectsource}) and (\ref{efm}).  \ We omit technical details of such proofs
because they are similar to the results for the Levi--Civita connection in
\cite{odints3} but (in our case) with distortions of formulas containing
covariant derivatives to be defined by the canonical d--connection, see (\ref%
{distrel}) and (\ref{distrel1}). All such constructions are metric
compatible and determined by the same metric structure. For $\ ^{ef}\mathbf{T%
}_{\beta \delta }=0,$ such details can be found in Refs.  \cite%
{ijgmmp,vexsol1,vexsol2} \ Effective source from modified gravity do no
change the N--adapted variational calcus. $\square $
\end{proof}

\vskip5pt

We consider matter field sources in (\ref{cdeinst}) which can be
diagonalized with respect to N--adapted frames, {\small
\begin{equation}
\mathbf{\Upsilon }_{~\delta }^{\beta }=diag[\mathbf{\Upsilon }_{\alpha }:%
\mathbf{\Upsilon }_{~1}^{1}=\mathbf{\Upsilon }_{~2}^{2}=\Upsilon
(x^{k},y^{3})+\underline{\Upsilon }(x^{k},y^{4});\mathbf{\Upsilon }_{~3}^{3}=%
\mathbf{\Upsilon }_{~4}^{4}=~^{v}\Upsilon (x^{k})].  \label{source}
\end{equation}%
} It can be performed via frame/ coordinate transforms for very general
distributions of matter fields. Such effective sources can be considered as
nonholonomic constraints via corresponding classes of $f$--functionals (\ref%
{functs}) on the Ricci tensor (see Theorem \ref{th2a}) and certain classes
of computed for certain general assumptions on modified off--diagonal
gravitational interactions.

\begin{corollary}
The gravitational field equations (\ref{cdeinst}) transform into the
Einstein equations in GR, in \textquotedblright standard\textquotedblright\
form for $\nabla ,$
\begin{equation}
E_{\beta \delta }=R_{\beta \delta }-\frac{1}{2}\mathbf{g}_{\beta \delta }\
R=\varkappa \ ^{m}T_{\beta \delta },  \label{einsteqs}
\end{equation}%
where $R:=\mathbf{g}^{\beta \delta }R_{\ \beta \delta },$ if $\ ^{2}f=0,$
and $\ ^{1}f(\ ^{s}\widehat{R})=R,$ for the same N--adapted coefficients for
both $\widehat{\mathbf{D}}$ and $\nabla $ if\
\begin{equation}
\widehat{L}_{aj}^{c}=e_{a}(N_{j}^{c}),\ \widehat{C}_{jb}^{i}=0,\ \Omega _{\
ji}^{a}=0,  \label{lcconstr}
\end{equation}
\end{corollary}

\begin{proof}
The systems of PDE (\ref{cdeinst}) and (\ref{einsteqs}) are very different.
But if the constraints (\ref{lcconstr}) are imposed additionally on $%
\widehat{\mathbf{D}},$ we satisfy the conditions of Proposition \ref{approp}%
, when the equalities $\Gamma _{\ \alpha \beta }^{\gamma }=\widehat{\mathbf{%
\Gamma }}_{\ \alpha \beta }^{\gamma }$ result in $R_{\beta \delta }=\widehat{%
\mathbf{R}}_{\ \beta \delta }$ and $E_{\alpha \beta }=\widehat{\mathbf{E}}%
_{\alpha \beta }.$ $\square $
\end{proof}

\vskip5pt

Finally, we note that the effective source $\mathbf{\Upsilon }_{~\delta
}^{\beta }$ and the canonical d--connection $\widehat{\mathbf{D}}$ encode
all information on modifications of the GR theory. Prescribing such values
following certain geometric/physical assumptions, we can model modified
gravity effects via generic off--diagonal solutions in GR if we are able to
construct such metrics via nonholonomic constraints $\widehat{\mathbf{D}}_{%
\mathcal{T}=0}\to\nabla $.

\subsection{Decoupling of modified gravitational field eqs}

\subsubsection{Off--diagonal metrics with Killing symmetry}

Let us consider an ansatz (\ref{ans1}) with $\omega =1,\underline{h}_{3}=1,%
\underline{w}_{i}=0$ and $\underline{n}_{i}=0$ for the data (\ref{paramdcoef}%
) and $\underline{\mathbf{\Upsilon }}=0$ for (\ref{source}). Such a generic
off--diagonal metric does not depend on variable $y^{4},$ i.e. $\partial
/\partial y^{4}$ is a Killing vector, if $\underline{h}_{4}=1.$
Nevertheless, the decoupling property can be proven for the same assumptions
but arbitrary $\underline{h}_{4}(x^{k},y^{4})$ with nontrivial dependence on
$y^{4}.$ We call this class of metrics to be with effective Killing symmetry
because they result in systems of PDE (\ref{cdeinst}) as for the Killing
case but there are differences in (\ref{lcconstr}) if $\underline{h}_{4}\neq
1.$ In order to simplify formulas, there will be used brief denotations for
partial derivatives,
\begin{equation*}
a^{\bullet }=\partial _{1}=\partial a/\partial x^{1},a^{\prime }=\partial
_{2}=\partial a/\partial x^{2},a^{\ast }=\partial _{3}=\partial a/\partial
y^{3},a^{\circ }=\partial _{4}=\partial a/\partial y^{4}.
\end{equation*}

\begin{theorem}
\label{th2a}The effective Einstein eqs (\ref{cdeinst}) and nonholonomic
constraints(\ref{lcconstr}) for a metric $\mathbf{g}$ (\ref{paramdcoef})
with $\omega =\underline{h}_{3}=1$ and $\underline{w}_{i}=\underline{n}%
_{i}=0 $ and $\underline{\mathbf{\Upsilon }}=0$ in matter source $\mathbf{%
\Upsilon }_{~\delta }^{\beta }$ (\ref{source}) are equivalent, respectively,
to {\small
\begin{eqnarray}
\widehat{R}_{1}^{1} &=&\widehat{R}_{2}^{2}=\frac{-1}{2g_{1}g_{2}}%
[g_{2}^{\bullet \bullet }-\frac{g_{1}^{\bullet }g_{2}^{\bullet }}{2g_{1}}-%
\frac{\left( g_{2}^{\bullet }\right) ^{2}}{2g_{2}}+g_{1}^{\prime \prime }-%
\frac{g_{1}^{\prime }g_{2}^{\prime }}{2g_{2}}-\frac{(g_{1}^{\prime })^{2}}{%
2g_{1}}]=-\ ^{v}\Upsilon ,  \label{eq1b} \\
\widehat{R}_{3}^{3} &=&\widehat{R}_{4}^{4}=-\frac{1}{2h_{3}h_{4}}%
[h_{4}^{\ast \ast }-\frac{\left( h_{4}^{\ast }\right) ^{2}}{2h_{4}}-\frac{%
h_{3}^{\ast }h_{4}^{\ast }}{2h_{3}}]=-\Upsilon ,  \label{eq2b} \\
\widehat{R}_{3k} &=&\frac{w_{k}}{2h_{4}}[h_{4}^{\ast \ast }-\frac{\left(
h_{4}^{\ast }\right) ^{2}}{2h_{4}}-\frac{h_{3}^{\ast }h_{4}^{\ast }}{2h_{3}}%
]+\frac{h_{4}^{\ast }}{4h_{4}}\left( \frac{\partial _{k}h_{3}}{h_{3}}+\frac{%
\partial _{k}h_{4}}{h_{4}}\right) -\frac{\partial _{k}h_{4}^{\ast }}{2h_{4}}%
=0,  \label{eq3b} \\
\widehat{R}_{4k} &=&\frac{h_{4}}{2h_{3}}n_{k}^{\ast \ast }+\left( \frac{h_{4}%
}{h_{3}}h_{3}^{\ast }-\frac{3}{2}h_{4}^{\ast }\right) \frac{n_{k}^{\ast }}{%
2h_{3}}=0,  \label{eq4b}
\end{eqnarray}%
}and{\small \
\begin{eqnarray}
w_{i}^{\ast } &=&(\partial _{i}-w_{i}\partial _{3})\ln |h_{3}|,(\partial
_{i}-w_{i}\partial _{3})\ln \sqrt{|h_{4}|}=0,\partial _{k}w_{i}=\partial
_{i}w_{k},  \label{lccond} \\
n_{k}\underline{h}_{4}^{\circ } &=&\mathbf{\partial }_{k}\underline{h}%
_{4},n_{i}^{\ast }=0,\partial _{i}n_{k}=\partial _{k}n_{i}.  \nonumber
\end{eqnarray}%
}
\end{theorem}

\begin{proof}
See a sketch of proof in Appendix \ref{sb}; more details and computations
for nonholonomic configurations in higher dimensions and the Einstein
gravity are provided in Refs. \cite{ijgmmp,vexsol2,veymh}. $\square $
\end{proof}

\vskip5pt

The system of PDE (\ref{eq1b})--(\ref{eq4b}) possess a very important
decoupling property which is characteristic for various classes of modified
gravity theories. Let us explain in brief this property. The equation (\ref%
{eq1b}) is for a 2--d metric which is always conformally flat and can be
diagonalized. Choosing any value of a function $g_{1}$ for a prescribed
source $\ ^{v}\Upsilon ,$ we can find $g_{2},$ or inversely. The equation (%
\ref{eq2b}) contains only the first and second derivatives on $\partial
/\partial y^{3}$ and relates two functions $h_{3}$ and $h_{4}.$ The
equations (\ref{eq3b}) consist a linear algebraic system for $w_{k}$ if the
coefficients $h_{a}$ have been already defined as a solution of (\ref{eq2b}%
). Nevertheless, we have to consider additional constraints on $w_{i}$ and $%
h_{4}$ solving a system of first order PDE on $x^{k}$ and $y^{3}$ in order
to find $w_{k}$ resulting in zero torsion conditions (\ref{lccond}). There
are additional conditions on $n_{k}.$ We shall analyze how the Levi--Civita
conditions can be solved in very general forms in section \ref{s3}.

\begin{corollary}
\label{corol1} The effective Einstein eqs (\ref{cdeinst}) and (\ref{lcconstr}%
) for a metric $\mathbf{g}$ (\ref{paramdcoef}) with $\omega =h_{4}=1$ and $%
w_{i}=n_{i}=0$ and $\mathbf{\Upsilon }=0$ in matter source $\mathbf{\Upsilon
}_{~\delta }^{\beta }$ (\ref{source}), when such values do not depend on
coordinate $y^{3}$ and posses one Killing symmetry on $\partial /\partial
y_{3},$ are equivalent, respectively, to {\small
\begin{eqnarray}
-\widehat{R}_{1}^{1} &=&-\widehat{R}_{2}^{2}=g_{2}^{\bullet \bullet }-\frac{%
g_{1}^{\bullet }g_{2}^{\bullet }}{2g_{1}}-\frac{\left( g_{2}^{\bullet
}\right) ^{2}}{2g_{2}}+g_{1}^{\prime \prime }-\frac{g_{1}^{\prime
}g_{2}^{\prime }}{2g_{2}}-\frac{(g_{1}^{\prime })^{2}}{2g_{1}}=2g_{1}g_{2}\
^{v}\Upsilon ,  \label{eq1c} \\
\widehat{R}_{3}^{3} &=&\widehat{R}_{4}^{4}=-\frac{1}{2\underline{h}_{3}%
\underline{h}_{4}}[\underline{h}_{3}^{\circ \circ }-\frac{\left( \underline{h%
}_{3}^{\circ }\right) ^{2}}{2\underline{h}_{3}}-\frac{\underline{h}%
_{3}^{\circ }\underline{h}_{4}^{\circ }}{2\underline{h}_{4}}]=-\underline{%
\Upsilon },  \label{eq2c} \\
\widehat{R}_{3k} &=&+\frac{\underline{h}_{3}}{2\underline{h}_{4}}\underline{w%
}_{k}^{\circ \circ }+\left( \frac{\underline{h}_{3}}{\underline{h}_{4}}%
\underline{h}_{4}^{\circ }-\frac{3}{2}\underline{h}_{3}^{\circ }\right)
\frac{\underline{h}_{k}^{\circ }}{2\underline{h}_{4}}=0,  \label{eq3c} \\
\widehat{R}_{4k} &=&\frac{\underline{n}_{k}}{2\underline{h}_{3}}[\underline{h%
}_{3}^{\circ \circ }-\frac{\left( \underline{h}_{3}^{\circ }\right) ^{2}}{2%
\underline{h}_{3}}-\frac{\underline{h}_{3}^{\circ }\underline{h}_{4}^{\circ }%
}{2\underline{h}_{4}}]+\frac{\underline{h}_{3}^{\circ }}{4\underline{h}_{3}}%
\left( \frac{\partial _{k}\underline{h}_{3}}{\underline{h}_{3}}+\frac{%
\partial _{k}\underline{h}_{4}}{\underline{h}_{4}}\right) -\frac{\partial
_{k}h_{3}^{\circ }}{2h_{3}}=0,  \label{eq4c}
\end{eqnarray}%
\begin{eqnarray}
\mbox{ \ and \ }\underline{n}_{i}^{\circ } &=&(\partial _{i}-\underline{n}%
_{i}\partial _{4})\ln |\underline{h}_{4}|,(\partial _{i}-\underline{n}%
_{i}\partial _{4})\ln |\underline{h}_{3}|=0,,  \label{lccondd} \\
(\partial _{k}-\underline{n}_{k}\partial _{4})\underline{n}_{i} &=&(\partial
_{i}-\underline{n}_{i}\partial _{4})\underline{n}_{k},\underline{w}%
_{k}h_{3}^{\ast }=\mathbf{\partial }_{k}h_{3},\underline{w}_{i}^{\circ
}=0,\partial _{i}\underline{w}_{k}=\partial _{k}\underline{w}_{i}.  \nonumber
\end{eqnarray}%
}
\end{corollary}

\begin{proof}
It is similar to that for Theorem \ref{th2a} but for $y^{3}\rightarrow
y^{4}. $ We do not repeat such computations. $\square $
\end{proof}

\vskip5pt

The nonlinear systems of PDE corresponding to modified gravitational field
equations (\ref{cdeinst}) and (\ref{lcconstr}) for metrics $\mathbf{g}$ (\ref%
{paramdcoef}) with Killing symmetry on $\partial /\partial y_{4},$ when $%
\omega =\underline{h}_{3}=1$ and $\underline{w}_{i}=\underline{n}_{i}=0$ and
$\underline{\mathbf{\Upsilon }}=0$ in matter source $\mathbf{\Upsilon }%
_{~\delta }^{\beta }$ (\ref{source}), can be transformed into respective
systems of PDE for data with Killing symmetry on $\partial /\partial y_{3},$
when $\omega =h_{4}=1$ and $w_{i}=n_{i}=0$ and $\mathbf{\Upsilon }=0$, if $%
h_{3}(x^{i},y^{3})\rightarrow \underline{h}_{4}(x^{i},y^{4}),$ $%
h_{4}(x^{i},y^{3})\rightarrow \underline{h}_{3}(x^{i},y^{4}),$ $%
w_{k}(x^{i},y^{3})\rightarrow \underline{n}_{k}(x^{i},y^{4})$ and $%
n_{k}(x^{i},y^{3})\rightarrow \underline{w}_{k}(x^{i},y^{4}).$

There is a possibility to preserve a N--adapted decoupling under "vertical"
conform transforms.

\begin{lemma}
\label{lemma1}The modified gravitational equations (\ref{eq2b})--(\ref{eq4b}%
), do not change under a "vertical" conformal transform with nontrivial $%
\omega (x^{k},y^{a})$ to a d--metric (\ref{paramdcoef}) if there are
satisfied the conditions%
\begin{equation}
\partial _{k}\omega -w_{i}\omega ^{\ast }-n_{i}\omega ^{\circ }=0\mbox{ and }%
\widehat{T}_{kb}^{a}=0.  \label{conf2a}
\end{equation}
\end{lemma}

\begin{proof}
We do not repeat here such details provided in Refs. \cite{vexsol1,vexsol2}
for $\underline{h}_{4}=1$ because a nontrivial $\underline{h}_{4}$ does not
modify substantially the proof. \ The computations from Appendix \ref{sb}
should be performed for coefficients $g_{i}(x^{k}),g_{3}=h_{3},g_{4}=h_{4}%
\underline{h}_{4},N_{i}^{3}=w_{i},$ $N_{i}^{4}=n_{i}$ are generalized to a
nontrivial $\omega (x^{k},y^{a})$ with $~^{\omega }g_{3}=\omega ^{2}h_{3}$
and $~^{\omega }g_{4}=\omega ^{2}h_{4}\underline{h}_{4}.$ Using formulas (%
\ref{cdc}), (\ref{distrel1}), (\ref{dcurv}) and (\ref{driccic}), we get
certain distortion relations for the Ricci d--tensors (\ref{riccid}), $%
~^{\omega }\widehat{R}_{~b}^{a}=\widehat{R}_{~b}^{a}+~^{\omega }\widehat{Z}%
_{~b}^{a}$ and $~^{\omega }\widehat{R}_{bi}=\widehat{R}_{bi}=0,$ where $%
\widehat{R}_{~b}^{a}$ and $\widehat{R}_{bi}$ are those computed for $\omega
=1,$ i.e. (\ref{eq2b})--(\ref{eq4b}). The values $~^{\omega }\widehat{R}%
_{~b}^{a}$ and$~^{\omega }\widehat{Z}_{~b}^{a}$ are defined by a nontrivial $%
\omega $ and computed using the same formulas. We get that $~^{\omega }%
\widehat{Z}_{~b}^{a}=0$ if the conditions (\ref{conf2a}) are satisfied. $%
\square $
\end{proof}

\vskip5pt

The conditions of the Theorem \ref{th2a}, Corollary \ref{corol1} and Lemma %
\ref{lemma1} result in a prove that

\begin{theorem}
\label{th2b}Any d--metric
\begin{eqnarray}
\mathbf{g} &=&g_{i}(x^{k})dx^{i}\otimes dx^{i}+\omega
^{2}(x^{k},y^{a})\left( h_{3}\mathbf{e}^{3}\otimes \mathbf{e}^{3}+h_{4}\
\underline{h}_{4}\mathbf{e}^{4}\otimes \mathbf{e}^{4}\right) ,  \nonumber  \\
\mathbf{e}^{3} &=&dy^{3}+w_{i}(x^{k},y^{3})dx^{i},\ \mathbf{e}%
^{4}=dy^{4}+n_{i}(x^{k})dx^{i},  \label{class1}
\end{eqnarray}%
satisfying the PDE (\ref{eq1b})-- (\ref{lccond}) and $\partial _{k}\omega
-w_{i}\omega ^{\ast }-n_{i}\omega ^{\circ }=0$, or any d--metric
\begin{eqnarray}
\mathbf{g} &=&g_{i}(x^{k})dx^{i}\otimes dx^{i}+\omega
^{2}(x^{k},y^{a})\left( h_{3}\underline{h}_{3}\mathbf{e}^{3}\otimes \mathbf{e%
}^{3}+\underline{h}_{4}\mathbf{e}^{4}\otimes \mathbf{e}^{4}\right) ,  \nonumber
\\
\mathbf{e}^{3} &=&dy^{3}+\underline{w}_{i}(x^{k})dx^{i},\ \mathbf{e}%
^{4}=dy^{4}+\underline{n}_{i}(x^{k},y^{4})dx^{i},  \label{class2}
\end{eqnarray}%
satisfying the PDE (\ref{eq1c})-- (\ref{lccondd}) and $\partial _{k}\omega -%
\underline{w}_{i}\omega ^{\ast }-\underline{n}_{i}\omega ^{\circ }=0$,
define (in general, different) two classes of generic off--diagonal
solutions \ of modified gravitational equations (\ref{cdeinst}) and (\ref%
{lcconstr}) with respective sources of type (\ref{source}).
\end{theorem}

Both ansatz of type (\ref{class1}) and (\ref{class2}) consist particular
cases of parameterizations of metrics in the form (\ref{ans1}). Via
frame/coordinate transform into a finite region of a point $~^{0}u\in
\mathbf{V}$ any (pseudo) Riemannian metric can be represented in an above
mentioned d--metric form. On Lorentz manifolds, only one of coordinates $%
y^{a}$ is timelike, i.e. the solutions of type (\ref{class1}) and (\ref%
{class2}) can not be transformed mutually via nonholonomic frame
deformations preserving causality.

\subsubsection{Effective linearization of Ricci tensors}

We can consider such local coordinates on an open region $U\subset \mathbf{V}
$ when computing the N--adapted coefficients of the Riemann and Ricci
d--tensors, see formulas (\ref{dcurv}) and (\ref{driccic}), we can neglect
contributions from quadratic terms of type $\widehat{\Gamma }\cdot \widehat{%
\Gamma }$ (preserving values of type $\partial _{\mu }\widehat{\Gamma }).$
These are N--adapted \ analogs of normal coordinates when $\widehat{\Gamma }%
(u_{0})=0$ for points $u_{0}$, for instance, belonging to a line on $U.$
Such conditions can be satisfied for decompositions of metrics and
connections on a small parameter like it is explained in details in Ref.
\cite{ijgmmp} (we shall consider decompositions on a small eccentricity
parameter $\varepsilon ,$ for ellipsoid configurations, in Section \ref{s5}%
). Other possibilities can be found if we impose nonholonomic constraints,
for instance, of type $h_{4}^{\ast }=0$ but for nonzero $h_{4}(x^{k},y^{3})$
and/or $h_{4}^{\ast \ast }(x^{k},y^{3}).$ Such constraints can be solved in
non--explicit form and define a corresponding subclass of N--adapted frames.
Considering additional nonholonomic deformations with a general decoupling
with respect to a \textquotedblright convenient\textquotedblright\ system of
reference/coordinates, we can deform the equations and solutions to
configurations with contributions from $\widehat{\Gamma }\cdot \widehat{%
\Gamma }$ terms.

\begin{theorem}[effective linearized decoupling]
\label{th2}The modified gravitational equations (\ref{cdeinst}) and (\ref%
{lcconstr}), via nonholonomic frame deformations to a metric $\mathbf{g}$ (%
\ref{paramdcoef}) and matter source $\mathbf{\Upsilon }_{~\delta }^{\beta }$
(\ref{source}), can be considered for an open region $U\subset \mathbf{V}$
when the contributions from terms of type $\widehat{\Gamma }\cdot \widehat{%
\Gamma }$ are small and we obtain and effective system of PDE with
h--v--decoupling: {\small
\begin{equation}
\widehat{R}_{1}^{1}=\widehat{R}_{2}^{2}=\frac{-1}{2g_{1}g_{2}}%
[g_{2}^{\bullet \bullet }-\frac{g_{1}^{\bullet }g_{2}^{\bullet }}{2g_{1}}-%
\frac{\left( g_{2}^{\bullet }\right) ^{2}}{2g_{2}}+g_{1}^{\prime \prime }-%
\frac{g_{1}^{\prime }g_{2}^{\prime }}{2g_{2}}-\frac{(g_{1}^{\prime })^{2}}{%
2g_{1}}]=-\ ^{v}\Upsilon ,  \label{eq1}
\end{equation}%
\begin{eqnarray}
\widehat{R}_{3}^{3} &=&\widehat{R}_{4}^{4}=-\frac{1}{2h_{3}h_{4}}%
[h_{4}^{\ast \ast }-\frac{\left( h_{4}^{\ast }\right) ^{2}}{2h_{4}}-\frac{%
h_{3}^{\ast }h_{4}^{\ast }}{2h_{3}}]-\frac{1}{2\underline{h}_{3}\underline{h}%
_{4}}[\underline{h}_{3}^{\circ \circ }-\frac{\left( \underline{h}_{3}^{\circ
}\right) ^{2}}{2\underline{h}_{3}}-\frac{\underline{h}_{3}^{\circ }%
\underline{h}_{4}^{\circ }}{2\underline{h}_{4}}]  \nonumber  \\
&=&-\Upsilon -\underline{\Upsilon },  \label{eq2} \\
\widehat{R}_{3k} &=&\frac{w_{k}}{2h_{4}}[h_{4}^{\ast \ast }-\frac{\left(
h_{4}^{\ast }\right) ^{2}}{2h_{4}}-\frac{h_{3}^{\ast }h_{4}^{\ast }}{2h_{3}}%
]+\frac{h_{4}^{\ast }}{4h_{4}}\left( \frac{\partial _{k}h_{3}}{h_{3}}+\frac{%
\partial _{k}h_{4}}{h_{4}}\right) -\frac{\partial _{k}h_{4}^{\ast }}{2h_{4}}
\label{eq3} \\
&&+\frac{\underline{h}_{3}}{2\underline{h}_{4}}\underline{n}_{k}^{\circ
\circ }+\left( \frac{\underline{h}_{3}}{\underline{h}_{4}}\underline{h}%
_{4}^{\circ }-\frac{3}{2}\underline{h}_{3}^{\circ }\right) \frac{\underline{n%
}_{k}^{\circ }}{2\underline{h}_{4}}=0,  \nonumber  \\
\widehat{R}_{4k} &=&\frac{\underline{w}_{k}}{2\underline{h}_{3}}[\underline{h%
}_{3}^{\circ \circ }-\frac{\left( \underline{h}_{3}^{\circ }\right) ^{2}}{2%
\underline{h}_{3}}-\frac{\underline{h}_{3}^{\circ }\underline{h}_{4}^{\circ }%
}{2\underline{h}_{4}}]+\frac{\underline{h}_{3}^{\circ }}{4\underline{h}_{3}}%
\left( \frac{\partial _{k}\underline{h}_{3}}{\underline{h}_{3}}+\frac{%
\partial _{k}\underline{h}_{4}}{\underline{h}_{4}}\right) -\frac{\partial
_{k}\underline{h}_{3}^{\circ }}{2\underline{h}_{3}}  \nonumber  \\
&&+\frac{h_{4}}{2h_{3}}n_{k}^{\ast \ast }+\left( \frac{h_{4}}{h_{3}}%
h_{3}^{\ast }-\frac{3}{2}h_{4}^{\ast }\right) \frac{n_{k}^{\ast }}{2h_{3}}=0,
\label{eq4}
\end{eqnarray}%
\begin{eqnarray}
w_{i}^{\ast } &=&(\partial _{i}-w_{i}\partial _{3})\ln |h_{3}|,(\partial
_{i}-w_{i}\partial _{3})\ln |h_{4}|=0  \nonumber  \\
(\partial _{k}-w_{k}\partial _{3})w_{i} &=&(\partial _{i}-w_{i}\partial
_{3})w_{k},n_{i}^{\ast }=0,\partial _{i}n_{k}=\partial _{k}n_{i},
\label{lcconstr1} \\
\underline{w}_{i}^{\circ } &=&0,\ \partial _{i}\underline{w}_{k}=\partial
_{k}\underline{w}_{i},(\partial _{k}-\underline{n}_{k}\partial _{4})%
\underline{n}_{i}=(\partial _{i}-\underline{n}_{i}\partial _{4})\underline{n}%
_{k},\   \nonumber  \\
\underline{n}_{i}^{\circ } &=&(\partial _{i}-\underline{n}_{i}\partial
_{4})\ln |\underline{h}_{4}|,(\partial _{i}-\underline{n}_{i}\partial
_{4})\ln |\underline{h}_{3}|=0  \nonumber  \\
\mathbf{e}_{k}\omega &=&\partial _{k}\omega -\left( w_{i}+\underline{w}%
_{i}\right) \omega ^{\ast }-\left( n_{i}+\underline{n}_{i}\right) \omega
^{\circ }=0.  \label{conf2}
\end{eqnarray}%
}
\end{theorem}

\begin{proof}
It follows from the Theorems \ref{th2a} and \ref{th2b} for any superposition
of ansatz (\ref{class1}) and (\ref{class2}) resulting into metrics of type (%
\ref{ans1}). $\square $
\end{proof}

\vskip5pt

In general, the solutions defined by a system (\ref{eq1})--(\ref{conf2}) can
not be transformed into solutions parameterized by an ansatz (\ref{class1})
and/or (\ref{class2}). As we shall prove in Section \ref{s3}, the general
solutions of the such systems of PDE are determined by corresponding sets of
generating and integration functions. A solution for (\ref{eq1})--(\ref%
{conf2}) contains a larger set of $h-v$--generating functions than those
with some N--coefficients stated to be zero.

\subsection{Decoupling of MGYMH equations}

We write for $\mathbf{T}_{\beta \delta }$ in (\ref{effectsource}) the
corresponding values for the energy--momentum tensors of the Yang--Mills,
YM, and Higgs, H, fields, when
\begin{equation*}
\mathbf{\Upsilon }_{\beta \delta }=8\pi \ ^{ef}\eta \ G\ \ ^{m}\mathbf{T}%
_{\beta \delta }+\ ^{ef}\mathbf{T}_{\beta \delta },
\end{equation*}%
where we introduced the coefficient $8\pi $ in order to get in the Einstein
limit solutions parameterized in the form \cite{bart,dzh,br1,br2}. A variational
N--adapted procedure can be elaborated on a manifold $\mathbf{V}$ when the
operator $\widehat{\mathbf{D}}$ is used instead of $\nabla $ and all
computations are performed with respect N--adapted bases (\ref{nader}) and (%
\ref{nadif}). It is completely similar to that for the Levi--Civita
connection but with N--elongated partial derivatives for a gravitating
non--Abbelian SU(2) gauge field $\mathbf{A}=\mathbf{A}_{\mu }\mathbf{e}^{\mu
}$ coupled to a triplet Higgs field $\Phi .$ We derive this system of
modified gravitational and matter field equations (in brief, MGYMH):%
\begin{eqnarray}
\widehat{\mathbf{R}}_{\ \beta \delta }-\frac{1}{2}\mathbf{g}_{\beta \delta
}\ ^{s}R &=&8\pi \ ^{ef}\eta G\left( \ ^{H}T_{\beta \delta }+\ ^{YM}T_{\beta
\delta }\right) +\ ^{ef}\mathbf{T}_{\beta \delta },  \label{ym1} \\
\ D_{\mu }(\sqrt{|g|}F^{\mu \nu }) &=&\frac{1}{2}ie\sqrt{|g|}[\Phi ,D^{\nu
}\Phi ],  \label{heq2} \\
D_{\mu }(\sqrt{|g|}\Phi ) &=&\lambda \sqrt{|g|}(\ \Phi _{\lbrack
0]}^{2}-\Phi ^{2})\Phi ,  \label{heq3}
\end{eqnarray}%
where the stress--energy tensors for the YM and H fields are computed
{\small
\begin{eqnarray}
\ ^{YM}T_{\beta \delta } &=&2Tr\left( \mathbf{g}^{\mu \nu }F_{\beta \mu
}F_{\delta \nu }-\frac{1}{4}\mathbf{g}_{\beta \delta }F_{\mu \nu }F^{\mu \nu
}\right) ,  \label{source1} \\
\ ^{H}T_{\beta \delta } &=&Tr[\frac{1}{4}\ (D_{\delta }\Phi \ D_{\beta }\Phi
+D_{\beta }\Phi \ D_{\delta }\Phi )-\frac{1}{4}\mathbf{g}_{\beta \delta
}D_{\alpha }\Phi \ D^{\alpha }\Phi ]-\mathbf{g}_{\beta \delta }\mathcal{V}%
(\Phi ).  \label{source2}
\end{eqnarray}%
} The value $\ ^{ef}\mathbf{T}_{\beta \delta }$ in (\ref{ym1}) is the same
as in (\ref{efm}) but for zero pressure, $p=0.$ The nonholonomic and
modified gravitational interactions of matter fields and the constants in (%
\ref{ym1})--(\ref{heq3}) are treated as follows: \ The non--Abbelian gauge
field with derivative $D_{\mu }=\mathbf{e}_{\mu }$ $+ie[\mathbf{A}_{\mu },]$
is changed into $\ \widehat{\mathcal{D}}_{\delta }=\widehat{\mathbf{D}}%
_{\delta }+ie[\mathbf{A}_{\delta },].$ The vector field $\mathbf{A}_{\delta }
$ is characterised by curvature
\begin{equation}
F_{\beta \mu }=\mathbf{e}_{\beta }\mathbf{A}_{\mu }-\mathbf{e}_{\mu }\mathbf{%
A}_{\beta }+ie[\mathbf{A}_{\beta },\mathbf{A}_{\mu }],  \label{gaugestr}
\end{equation}%
where $e$ is the coupling constant, $i^{2}=-1,$ and $[\cdot ,\cdot ]$ is
used for the commutator. We also consider that the value $\Phi _{\lbrack 0]}$
in (\ref{heq3}) is the vacuum expectation of the Higgs field which
determines the mass $\ ^{H}M=\sqrt{\lambda }\eta ;$ the value $\lambda $ is
the \ constant of scalar field self--interaction with potential $\mathcal{V}%
(\Phi )=\frac{1}{4}\lambda Tr(\Phi _{\lbrack 0]}^{2}-\Phi ^{2})^{2};$ the
gravitational constant $G$ defines the Plank mass $M_{Pl}=1/\sqrt{G}$ and it
is also the mass of a gauge boson, $\ ^{W}M=ev.$

Let us consider that a \textquotedblright prime\textquotedblright\ solution
is known for the system (\ref{ym1})--(\ref{heq3}) (given by data for a
diagonal d--metric $\ ^{\circ }\mathbf{g=}[\ ^{\circ }g_{i}(x^{1}),\ ^{\circ
}h_{a}(x^{k}),$ $\ ^{\circ }N_{i}^{a}=0]$ and matter fields $\ ^{\circ
}A_{\mu }(x^{1})$ and $^{\circ }\Phi (x^{1}),$ for instance, of type
constructed in Ref. \cite{br3} (see also Appendix \ref{aseymheq})). We
suppose that there are satisfied the following conditions:

\begin{enumerate}
\item The d--metric $~^{\eta }\mathbf{g}$ with nontrivial N--coefficients
for $\ ^{\circ }\mathbf{g\rightarrow }~^{\eta }\mathbf{g}$ is parameterized
by an ansats of type (\ref{ans1}), {\small
\begin{eqnarray}
\ \mathbf{g} &=&\eta _{i}(x^{k})\ ^{\circ }g_{i}(x^{1})dx^{i}\otimes
dx^{i}+\eta _{a}(x^{k},y^{a})\ ^{\circ }h_{a}(x^{1},x^{2})\mathbf{e}%
^{a}\otimes \mathbf{e}^{a}  \nonumber  \\
&=&g_{i}(x^{k})dx^{i}\otimes dx^{i}+\omega
^{2}(x^{k},y^{b})h_{a}(x^{k},y^{a})\mathbf{e}^{a}\otimes \mathbf{e}^{a},
\label{dm1} \\
\mathbf{e}^{3} &=&dy^{3}+[w_{i}+\underline{w}_{i}]dx^{i},\ \mathbf{e}%
^{4}=dy^{3}+[n_{i}+\underline{n}_{i}]dx^{i}.  \nonumber
\end{eqnarray}%
}

\item The non--Abbelian gauge fields are nonholonomically deformed as
\begin{equation}
A_{\mu }(x^{i},y^{3})=\ ^{\circ }A_{\mu }(x^{1})+\ ^{\eta }A_{\mu
}(x^{i},y^{a}),  \label{ans2a}
\end{equation}%
where $\ ^{\circ }A_{\mu }(x^{1})$ is defined by an ansatz (\ref{ans1a}) and
$\ ^{\eta }A_{\mu }(x^{i},y^{a})$ are any functions
\begin{equation}
F_{\beta \mu }=\ ^{\circ }F_{\beta \mu }(x^{1})+\ \ ^{\eta }F_{\beta \mu
}(x^{i},y^{a})=s\sqrt{|g|}\varepsilon _{\beta \mu },  \label{gaugstr1}
\end{equation}%
for $s=const$ and $\varepsilon _{\beta \mu }$ being the absolute
antisymmetric tensor. The gauge field curvatures $F_{\beta \mu },\ ^{\circ
}F_{\beta \mu }$ and $\ ^{\eta }F_{\beta \mu }$ are computed by introducing (%
\ref{ans1a}) and (\ref{ans2a}) into (\ref{gaugestr}). It should be
emphasized that an antisymmetric tensor $F_{\beta \mu }$ (\ref{gaugstr1})
solves the equations $D_{\mu }(\sqrt{|g|}F^{\mu \nu })=0;$ we can always
determine $\ ^{\eta }F_{\beta \mu }$ and $\ ^{\eta }A_{\mu },$ for any given
$\ ^{\circ }A_{\mu }$ and$\ ^{\circ }F_{\beta \mu }.$

\item The scalar field is nonholonomically modified by gravitational and
gauge field interactions $\ \ ^{\circ }\Phi (x^{1})\rightarrow \Phi
(x^{i},y^{a})=\ ^{\Phi }\eta (x^{i},y^{a})\ ^{\circ }\Phi (x^{1})$ by a
polarization $\ ^{\Phi }\eta $ is such way that
\begin{equation}
D_{\mu }\Phi =0\mbox{ and \ }\Phi (x^{i},y^{a})=\pm \Phi _{\lbrack 0]}.
\label{cond3}
\end{equation}%
Such nonholonomic modifications of the nonlinear scalar field is not trivial
even with respect to N--adapted frames $\mathcal{V}(\Phi )=0$ and $\
^{H}T_{\beta \delta }=0,$ see formula (\ref{source1}). For ansatz (\ref{dm1}%
), the equations (\ref{cond3}) transform into%
\begin{eqnarray}
(\partial /\partial x^{i}-A_{i})\Phi &=&(w_{i}+\underline{w}_{i})\Phi ^{\ast
}+(n_{i}+\underline{n}_{i})\Phi ^{\circ },  \label{cond2a} \\
\left( \partial /\partial y^{3}-A_{3}\right) \Phi &=&0,\ \left( \partial
/\partial y^{4}-A_{4}\right) \Phi =0.  \nonumber
\end{eqnarray}%
So, a nonolonomically constrained/deformed Higgs $\Phi $ field (depending in
non--explicit form on two variables because of constraint (\ref{cond3}))
modifies indirectly the off--diagonal components of the metric via $w_{i}+%
\underline{w}_{i}$ and $n_{i}+\underline{n}_{i}$ and conditions (\ref{cond2a}%
) for $\ ^{\eta }A_{\mu }.$ Such modifications can compensate $f$%
--modifications.

\item The non--Abbelian gauge fields (\ref{gaugstr1}) with the potential $%
A_{\mu }$ (\ref{ans2a}) modified nonholonomically by $\Phi $ subjected to
the conditions (\ref{cond3}) and with gravitational $f$--modifications
determine exact solutions of the system (\ref{eq2}) and (\ref{eq3}) if the
metric ansatz is chosen to be in the form (\ref{dm1}). The energy--momentum
tensor is computed\footnote{%
such a calculus in coordinate frames is provided in sections 3.2 and 6.51 in
Ref. \cite{lidsey}} $\ ^{YM}T_{\beta }^{\alpha }=-4s^{2}\delta _{\beta
}^{\alpha }m$ (see similar results in sections 3.2 and 6.51 in Ref. \cite%
{lidsey}). Such (modified) gravitationally interacting gauge and Higgs
fields, with respect to N--adapted frames, result in an effective
cosmological constant $\ ^{s}\lambda =8\pi s^{2}$ which should be added to a
respective source (\ref{source}).
\end{enumerate}

We conclude that an ansatz $\mathbf{g}=[\eta _{i}\ ^{\circ }g_{i},\eta _{a}\
^{\circ }h_{a};w_{i},n_{i}]$ (\ref{dm1}) and certain gauge--scalar
configurations $(A,\Phi )$ subjected to above mentioned conditions 1-4
define a decoupling of the system (\ref{ym1})--(\ref{heq3}) in a form stated
respectively by the Theorems \ref{th2a}, \ref{th2b}, and/or \ref{th2} if the
sources (\ref{source}) are redefined in the form
\begin{equation*}
\mathbf{\Upsilon }_{~\delta }^{\beta }=diag[\mathbf{\Upsilon }_{\alpha
}]\rightarrow \mathbf{\Upsilon }_{~\delta }^{\beta }+\ ^{YM}T_{~\delta
}^{\beta }=diag[\mathbf{\Upsilon }_{\alpha }-4s^{2}\delta _{\beta }^{\alpha
}].
\end{equation*}
In N--adapted frames the contributions of effective $f$--sources and matter
fields is defined by an effective cosmological constant $\ ^{s}\lambda .$

\section{Off--Diagonal Solutions for Modified Gravitational YMH Eqs}

\label{s3} In this section, we show how the decoupling property of the MGYMH
equations allows us to integrate such PDE in very general forms depending on
properties of coefficients of ansatz for metrics.

\subsection{Generating solutions with weak one Killing symmetry}

We prove that the MGYMH equations encoding gravitational and YMH
interactions and satisfying the conditions of Theorem \ref{th2a} can be
integrated in general forms for $h_{a}^{\ast }\neq 0$ and certain special
cases with zero and non--zero sources (\ref{source}). In general, such
generic off--diagonal metrics are determined by generating functions
depending on three/four coordinates.

\subsubsection{(Non) vacuum metrics with $h_{a}^{\ast }\neq 0$}

For ansatz (\ref{ans1}) with data $\omega =1,\underline{h}_{3}=1,\underline{w%
}_{i}=0$ and $\underline{n}_{i}=0$ for (\ref{paramdcoef}), when $h_{a}^{\ast
}\neq 0,$ and the condition that the source
\begin{equation}
\mathbf{\Upsilon }_{~\delta }^{\beta }=diag[\mathbf{\Upsilon }_{\alpha }:%
\mathbf{\Upsilon }_{~1}^{1}=\mathbf{\Upsilon }_{~2}^{2}=\Upsilon
(x^{k},y^{3})-4s^{2};\mathbf{\Upsilon }_{~3}^{3}=\mathbf{\Upsilon }%
_{~4}^{4}=~^{v}\Upsilon (x^{k})-4s^{2}],  \label{source1a}
\end{equation}%
is not zero, the solutions of Einstein eqs can be constructed following

\begin{theorem}
\label{th3a}The MGYMH equations (\ref{eq1b})-- (\ref{eq4b}) with source (\ref%
{source1a}) can be integrated in general forms by metrics
\begin{eqnarray}
\mathbf{g} &=&\epsilon _{i}e^{\psi (x^{k})}dx^{i}\otimes dx^{i}+\frac{|%
\tilde{\Theta}^{\ast }|^{2}}{\breve{\Upsilon}\tilde{\Theta}^{2}}\mathbf{e}%
^{3}\otimes \mathbf{e}^{3}-\frac{\tilde{\Theta}^{2}}{|\Lambda |}\underline{h}%
_{4}(x^{k},y^{4})\mathbf{e}^{4}\otimes \mathbf{e}^{4},  \nonumber  \\
\mathbf{e}^{3} &=&dy^{3}+\partial _{i}\widetilde{K}(x^{k},y^{3})dx^{i},\
\mathbf{e}^{4}=dy^{4}+\partial _{i}n(x^{k})dx^{i},  \label{sol1}
\end{eqnarray}%
with coefficients determined by generating functions $\psi (x^{k}),\tilde{%
\Theta}(x^{k},y^{3}),$ $\tilde{\Theta}^{\ast }\neq 0,$ $n_{i}(x^{k})$ and $%
\underline{h}_{4}(x^{k},y^{4}),$ and effective cosmological constant $%
\Lambda $ and source $\Upsilon -4s=\breve{\Upsilon}(x^{k})\neq 0$ following
recurrent formulas and conditions%
\begin{eqnarray}
\epsilon _{1}\psi ^{\bullet \bullet }+\epsilon _{2}\psi ^{\prime \prime }
&=&2\ [\ ^{v}\Upsilon -4s^{2}];  \label{sol1a} \\
h_{4} &=&-K^{2}=-\left[ \ _{0}K^{2}+\int dy^{3}\frac{(\Theta ^{2})^{\ast }}{%
4(\Upsilon -4s^{2})}\right]  \label{sol1b} \\
&=&(\ _{0}K^{2}+\tilde{\Theta}^{2}/\Lambda );  \nonumber  \\
h_{3} &=&B^{2}=4(K^{\ast })^{2}/\Theta ^{2}  \nonumber  \\
&=&(\Theta ^{\ast })^{2}/4(\Upsilon -4s^{2})^{2}\left[ \ _{0}K^{2}+\int
dy^{3}(\Theta ^{2})^{\ast }/4(\Upsilon -4s^{2})\right]  \nonumber  \\
&=&|\tilde{\Theta}^{\ast }|^{2}/\int dy^{3}(\Upsilon -4s^{2})(\tilde{\Theta}%
^{2})^{\ast }=|\tilde{\Theta}^{\ast }|^{2}/\breve{\Upsilon}\tilde{\Theta}%
^{2};  \label{sol1c} \\
w_{i} &=&\partial _{i}\phi /\phi ^{\ast }=\partial _{i}\Theta /\Theta ^{\ast
},  \label{sol1d} \\
&=&\partial _{i}K/K^{\ast }=\partial _{i}|\tilde{\Theta}|/|\tilde{\Theta}%
|^{\ast },\mbox{ if }\Upsilon -4s=\breve{\Upsilon}(x^{k}).  \nonumber
\end{eqnarray}%
where the constraints
\begin{eqnarray}
w_{i}^{\ast } &=&(\partial _{i}-w_{i}\partial _{3})\ln |h_{3}|,(\partial
_{i}-w_{i}\partial _{3})\ln |h_{4}|=0,  \label{lccondm} \\
\partial _{k}w_{i} &=&\partial _{i}w_{k},n_{k}\underline{h}_{4}^{\circ }=%
\mathbf{\partial }_{k}\underline{h}_{4},\partial _{i}n_{k}=\partial
_{k}n_{i}.  \nonumber
\end{eqnarray}%
are used for the Levi--Civita configurations.
\end{theorem}

\begin{proof}
We sketch a proof which related to similar ones in \cite%
{vexsol1,vexsol2,ijgmmp} if $\underline{h}_{4}=1.$ The N--adapted
coefficients of a metric (\ref{dm}) are parameterized in the form
\begin{equation}
g_{i}=e^{\psi {(x^{k})}},g_{a}=\omega (x^{k},y^{b})h_{a}(x^{k},y^{3}),\
N_{i}^{3}=w_{i}(x^{k},y^{3}),N_{i}^{4}=n_{i}(x^{k}),  \label{data2}
\end{equation}%
considering that for certain frame/coordinate transforms we can satisfy the
conditions $h_{a}^{\ast }\neq 0,\Upsilon _{2,4}\neq 0.$

\begin{itemize}
\item We introduce the functions
\begin{eqnarray}
\phi (x^{k},y^{3}) &:=&\ln \left\vert \frac{h_{4}^{\ast }}{\sqrt{|h_{3}h_{4}|%
}}\right\vert ,\ \Theta :=e^{{\phi }}  \label{auxcoef} \\
\gamma &:=&\left( \ln \frac{|h_{4}|^{3/2}}{|h_{3}|}\right) ^{\ast },\ \
\alpha _{i}=h_{4}^{\ast }\partial _{i}\phi ,\ \beta =h_{4}^{\ast }\phi
^{\ast }  \label{auxcoef1}
\end{eqnarray}%
The system of equations (\ref{eq1b})--(\ref{eq4b}) transforms into
\begin{eqnarray}
\psi ^{\bullet \bullet }+\psi ^{\prime \prime } &=&2[~^{v}\Upsilon -4s^{2}],
\label{eq1bb} \\
\phi ^{\ast }h_{4}^{\ast } &=&2h_{3}h_{4}\left[ \Upsilon -4s^{2}\right]
\label{eq2bb} \\
\beta w_{i}-\alpha _{i} &=&0,  \label{eq3bb} \\
n_{i}^{\ast \ast }+\gamma n_{i}^{\ast } &=&0,  \label{eq4bb} \\
\partial _{i}\omega -(\partial _{i}\phi /\phi ^{\ast })\omega ^{\ast
}-n_{i}\omega ^{\diamond } &=&0,  \label{confeq}
\end{eqnarray}

\item A horizontal metric $g_{i}(x^{2})$ is for a 2--d subspace and can be
represented in a conformally flat form $\epsilon _{i}e^{\psi
(x^{k})}dx^{i}\otimes dx^{i}.$ For such a h--metric, the equation (\ref%
{eq1bb}) is a 2-d Laplace equation which can be solved exactly if a source $%
^{v}\Upsilon (x^{k})-4s^{2}$ is prescribed from formulas (\ref{source}).

\item For un--known two functions $K=+\sqrt{|h_{4}|},B=+\sqrt{|h_{3}|},$ the
system of equations (\ref{auxcoef}) and (\ref{eq2bb}) can be written in the
form
\begin{eqnarray}
\Theta ^{\ast }K^{\ast } &=&(\Upsilon -4s^{2})B^{2}\Theta K,  \label{au1} \\
\Theta B &=&2K^{\ast }.  \label{au2}
\end{eqnarray}%
We introduce $\Theta B$ from (\ref{au2}) in the right of (\ref{au1}) for $%
A^{\ast }\neq 0,$ and find $2AB\Upsilon =\Phi ^{\ast }$ and devide to (\ref%
{au2}) with nonzero coefficients. We express
\begin{equation}
(K^{2})^{\ast }=(\Theta ^{2})^{\ast }/4(\Upsilon -4s^{2}).  \label{au3}
\end{equation}%
Integrating on $y^{3},$ we find (\ref{sol1b}) for an integration function $\
_{0}K(x^{k}),$ which can be included in $\tilde{\Theta},$ and $\epsilon
_{4}=\pm 1$ depending on signature of the metric. We introduced in that
formula an effective cosmological constant $\Lambda $ and re--defined the
generating function, $\Theta \rightarrow \tilde{\Theta},$
\begin{equation}
\frac{(\Theta ^{2})^{\ast }}{4(\Upsilon -4s^{2})}=\frac{(\tilde{\Theta}%
^{2})^{\ast }}{\Lambda },  \label{auxx4}
\end{equation}%
choosing $K=\tilde{K}=|\tilde{\Theta}|/\sqrt{|\Lambda |}.$ Using (\ref{au2}%
), (\ref{sol1b}) and (\ref{au3}), rewritten in the form $K^{\ast }K=\Theta
^{\ast }\Theta /4(\Upsilon -4s^{2}),$ we obtain $B=2K^{\ast }/\Theta ,$ i.e.
(\ref{sol1c}). That formula is written for $\Upsilon -4s^{2}=\breve{\Upsilon}%
(x^{k})$ when we can transform $h_{a}[\ _{0}K,\Theta ,\Upsilon
-4s^{2}]\rightarrow h_{a}[\tilde{\Theta},\Lambda ],$ for off--diagonal
configurations determined by a cosmological constant $\Lambda $ and
generating function $\tilde{\Theta}(x^{k},y^{3}).$

\item The algebraic equations for $w_{i}$ can be solved by introducing the
coefficients (\ref{auxcoef1}) in (\ref{eq3bb}) for the generating function $%
\phi ,$ or using equivalent variables,
\begin{eqnarray}
w_{i} &=&\partial _{i}\phi /\phi ^{\ast }=\partial _{i}\Theta /\Theta ^{\ast
}  \label{w1b} \\
&=&\partial _{i}K/K^{\ast }=\partial _{i}|\tilde{\Theta}|/|\tilde{\Theta}%
|^{\ast },\mbox{ if }\Upsilon -4s^{2}=\breve{\Upsilon}(x^{k}).  \label{w1bl}
\end{eqnarray}

\item Integrating two times on $y^{3},$ we obtain the solution of (\ref%
{eq4bb}):
\begin{equation*}
n_{k}=\ _{1}n_{k}+\ _{2}n_{k}\int dy^{3}\ h_{3}/(\sqrt{|h_{4}|})^{3}=\
_{1}n_{k}+\ _{2}n_{k}\int dy^{3}K^{2}/B^{3},
\end{equation*}%
for integration functions $\ _{1}n_{k}(x^{i}),\ _{2}n_{k}(x^{i}).$

\item The nonholonomic Levi--Civita conditions (\ref{lccondm}) can not be
solved in explicit form for arbitrary data $(K,\Upsilon -4s^{2}),$ or $%
(\Theta ,\Upsilon -4s^{2}),$ and arbitrary integration functions $\
_{1}n_{k} $ and $\ _{2}n_{k};$ we can fix $\ _{2}n_{k}=0$ and $\
_{1}n_{k}=\partial _{k}n$ with a function $n=n(x^{k}).$ We emphasize that $%
(\partial _{i}-w_{i}\partial _{3})\Theta \equiv 0$ for any $\Theta
(x^{k},y^{3})$ if $w_{i}$ is computed following formula (\ref{w1b}).
Introducing a new functional $H(\Theta )$ instead of $\Theta ,$ we obtain $%
(\partial _{i}-w_{i}\partial _{3})H=\frac{\partial H}{\partial \Theta }%
(\partial _{i}-w_{i}\partial _{3})\Theta =0$. Any formula (\ref{sol1b}) for
functionals of type $h_{4}=H(|\tilde{\Theta}(\Theta )|),$ we solve always
the equations $(\partial _{i}-w_{i}\partial _{3})h_{4}=0,$ which is
equivalent to the second system of equations in (\ref{lccondm}) because $%
(\partial _{i}-w_{i}\partial _{3})\ln \sqrt{|h_{4}|}\sim (\partial
_{i}-w_{i}\partial _{3})h_{4}.$ We compute for the left part of the second
equation, $\ (\partial _{i}-w_{i}\partial _{3})\ln \sqrt{|h_{4}|}=0,$ for a
subclass of generating functions $\Theta =\check{\Theta}$ for which
\begin{equation}
(\partial _{i}\check{\Theta})^{\ast }=\partial _{i}\check{\Theta}^{\ast }
\label{aux4a}
\end{equation}%
and (\ref{sol1b}). The first system of equations in (\ref{lccondm}) are
solved in explicit form if $w_{i}$ are determined by formulas (\ref{w1bl}),
and $h_{3}[\tilde{\Theta}]$ and $h_{4}[\tilde{\Theta},\tilde{\Theta}^{\ast
}] $ are \ respectively for (\ref{sol1b}) and (\ref{sol1c}) when $\Upsilon
-4s=\breve{\Upsilon}(x^{k}).$ We can write write the formulas
\begin{equation*}
w_{i}=\partial _{i}|\tilde{\Theta}|/|\tilde{\Theta}|^{\ast }=\partial
_{i}|\ln \sqrt{|h_{3}|}|/|\ln \sqrt{|h_{3}|}|^{\ast }
\end{equation*}%
if $\tilde{\Theta}=\tilde{\Theta}(\ln \sqrt{|h_{3}|})$ and $h_{3}[\tilde{%
\Theta}[\check{\Theta}]].$ Taking derivative $\partial _{3}$ on both sides
of previous equation, we compute
\begin{equation*}
w_{i}^{\ast }=\frac{(\partial _{i}|\ln \sqrt{|h_{3}|}|)^{\ast }}{|\ln \sqrt{%
|h_{3}|}|^{\ast }}-w_{i}\frac{|\ln \sqrt{|h_{3}|}|^{\ast \ast }}{|\ln \sqrt{%
|h_{3}|}|^{\ast }}.
\end{equation*}%
This way we are able to construct generic off--diagonal configurations with $%
w_{i}^{\ast }=(\partial _{i}-w_{i}\partial _{3})\ln \sqrt{|h_{3}|},$ which
is necessary for zero torsion conditions, if the constraints (\ref{aux4a})
are imposed. The conditions $\partial _{k}w_{i}=\partial _{i}w_{k}$ from the
second line in (\ref{lccondm}) $\ $are satisfied by any
\begin{equation}
\check{w}_{i}=\partial _{i}\check{\Theta}/\check{\Theta}^{\ast }=\partial
_{i}\widetilde{K},  \label{w1c}
\end{equation}%
when a nontrivial $\widetilde{K}(x^{k},y^{3})$ exists. $\square $
\end{itemize}
\end{proof}

\vskip5pt

The solutions constructed in Theorem \ref{th3a}, and those derived following
Corollary \ref{corol1} are very general ones and contain as particular cases
all known exact solutions for (non) holonomic Einstein spaces with Killing
symmetries and $f$--modifications studied in this paper. They can be
generalized to include arbitrary finite sets of parameters, see \cite{ijgmmp}%
.

For arbitrary $K$ and $\Upsilon -4s,$ and related $\Theta ,$ or $\tilde{%
\Theta},$ and $\Lambda ,$ we can generate off--diagonal solutions of (\ref%
{eq1b})--(\ref{eq4b}) with nonholonomically induced torsion completely
determined by the metric structure,
\begin{eqnarray}
ds^{2} &=&e^{\psi (x^{k})}[(dx^{1})^{2}+(dx^{2})^{2}]+K^{2}[dy^{3}+\frac{%
\partial _{i}\Theta }{\Theta ^{\ast }}dx^{i}]^{2}  \label{qelgent} \\
&&-B^{2}[dt+(\ _{1}n_{k}+\ _{2}n_{k}\int dy^{3}K^{2}/B^{3})dx^{k}]^{2},
\nonumber
\end{eqnarray}%
\ where the generating functions $K,B$ and $\Theta $ are related via
formulas (\ref{sol1b}) and (\ref{sol1c}) but not subjected to the conditions
(\ref{lccondm}).

\subsubsection{Off--diagonal effective vacuum EYMH configurations}

We can consider a subclass of generic off--diagonal MGYMH interactions which
can be encoded as effective vacauum Einstein manifolds when $\Upsilon
=4s^{2}.$ In general, such classes of solutions depend parametrically on $%
\Upsilon -4s^{2}$ and do not have a smooth limit from non-vacuum to vacuum
models.

\begin{corollary}
\label{corolvacuum}The effective vacuum solutions for the EYHM systems with
ansatz for metrics of type (\ref{sol1}) with vanishing source (\ref{source1a}%
) are parameter\-iz\-ed in the form
\begin{eqnarray}
\mathbf{g} &=&\epsilon _{i}e^{\psi (x^{k})}dx^{i}\otimes
dx^{i}+h_{3}(x^{k},y^{3})\mathbf{e}^{3}\otimes \mathbf{e}%
^{3}+h_{4}(x^{k},y^{3})\underline{h}_{4}(x^{k},y^{4})\mathbf{e}^{4}\otimes
\mathbf{e}^{4},  \nonumber  \\
\mathbf{e}^{3} &=&dy^{3}+w_{i}(x^{k},y^{3})dx^{i},\ \mathbf{e}%
^{4}=dy^{4}+n_{i}(x^{k})dx^{i},  \label{sol2}
\end{eqnarray}%
where coefficients are defined by solutions of the system
\begin{eqnarray}
\ddot{\psi}+\psi ^{\prime \prime } &=&0,  \label{ep1a} \\
\phi ^{\ast }\ h_{4}^{\ast } &=&0,  \label{ep2a} \\
\beta w_{i}+\alpha _{i} &=&0,  \label{ep3a}
\end{eqnarray}%
where the coefficients are subjected additionally to the zero--torsion
conditions (\ref{lccondm}).
\end{corollary}

\begin{proof}
The equations (\ref{ep1a})-- (\ref{ep3a}) are, respectively, (\ref{eq1bb})--
(\ref{eq3bb}) with zero sources. \ To solve (\ref{ep1a}) we can take $\psi
=0 $ or consider a trivial 2-d wave equation if one of coordinates $x^{k}$
is timelike.

There are two classes of solutions for (\ref{ep2a}):

The first one is to consider that $h_{4}=h_{4}(x^{k}),$ i.e. $h_{4}^{\ast
}=0,$ which states that the equation (\ref{ep2a}) has solutions for
arbitrary function $h_{3}(x^{k},y^{3})$ and arbitrary N--coefficients $%
w_{i}(x^{k},y^{3}),$ see (\ref{auxcoef1}). The functions $h_{3}$ and $w_{i}$
can be taken as generation ones which should be constrained only by the
conditions (\ref{lccondm}). The second equations in such conditions
constrain substantially the class of admissible $w_{i}$ if $h_{4}$ depends
only on $x^{k}.$ Nevertheless, $h_{3}$ can be an arbitrary one generating
solutions which can be extended for nontrivial sources $\underline{\Upsilon }
$ and systems (\ref{eq1c})-- (\ref{lccondd}) and/or (\ref{eq1})-- (\ref%
{conf2}).

The second class of solutions can be generated after corresponding
coordinate transforms, $\phi =\ln \left\vert h_{4}^{\ast }/\sqrt{|h_{3}h_{4}|%
}\right\vert =\ ^{0}\phi =const,~\phi ^{\ast }=0$ and $h_{4}^{\ast }\neq 0.$
We can solve (\ref{ep2a}) if
\begin{equation}
\sqrt{|h_{3}|}=\ ^{0}h(\sqrt{|h_{4}|})^{\ast },  \label{rel1}
\end{equation}%
for $\ ^{0}h=const\neq 0.$ Such v--metrics are generated by any function $%
\varpi (x^{i},y^{3}),$\ with $\varpi ^{\ast }\neq 0,$ when
\begin{equation}
h_{4}=-\varpi ^{2}\left( x^{i},y^{3}\right) \mbox{ and }h_{3}=(\ ^{0}h)^{2}\ %
\left[ \varpi ^{\ast }\left( x^{i},y^{3}\right) \right] ^{2};  \label{aux2}
\end{equation}%
for $N_{i}^{a}\rightarrow 0$ we obtain diagonal metrics with signature $%
(+,+,+,-).$ The coefficients $\alpha _{i}=\beta =0$ in (\ref{ep3a}) and $%
w_{i}(x^{k},y^{3})$ can be any functions subjected to the conditions (\ref%
{lccondm}), or equivalently to
\begin{eqnarray}
w_{i}^{\ast } &=&2\mathbf{\partial }_{i}\ln |\varpi |-2w_{i}(\ln |\varpi
|)^{\ast },  \label{cond1} \\
\partial _{k}w_{i}-\mathbf{\partial }_{i}w_{k} &=&2(w_{k}\partial
_{i}-w_{i}\partial _{k})\ln |\varpi |,  \nonumber
\end{eqnarray}%
for any $n_{i}(x^{k})$ when $\partial _{i}n_{k}=\partial _{k}n_{i}$.
Constraints of type $n_{k}\underline{h}_{4}^{\circ }=\mathbf{\partial }_{k}%
\underline{h}_{4}$ (\ref{aux4}) have to be imposed for a nontrivial multiple
$\underline{h}_{4}.$ $\square $
\end{proof}

\vskip5pt

Using Corollary \ref{corol1}, the "dual" ansatz to (\ref{sol2}) with $%
y^{3}\rightarrow y^{4}$ and $y^{4}\rightarrow y^{3}$ can be used to generate
effective vacuum solutions with weak Killing symmetry on $\partial /\partial
y^{3}.$

\subsection{Effective EYMH configurations with non--Killing symmetries}

The Theorem \ref{th2b} can be applied for constructing non--vacuum and
effective vacuum solutions of the EYMH equations depending on all
coordinates without explicit Killing symmetries.

\subsubsection{Non--vacuum off--diagonal solutions}

We can generate such YMH Einstein manifolds following

\begin{corollary}
\label{corym1}An ansatz of type (\ref{class1}) with a d--metric
\begin{eqnarray*}
\mathbf{g} &=&\epsilon _{i}e^{\psi (x^{k})}dx^{i}\otimes dx^{i}+\omega ^{2}[%
\frac{|\tilde{\Theta}^{\ast }|^{2}}{\breve{\Upsilon}\tilde{\Theta}^{2}}%
\mathbf{e}^{3}\otimes \mathbf{e}^{3}-\frac{\tilde{\Theta}^{2}}{|\Lambda |}%
\underline{h}_{4}(x^{k},y^{4})\mathbf{e}^{4}\otimes \mathbf{e}^{4}], \\
\mathbf{e}^{3} &=&dy^{3}+\partial _{i}\widetilde{K}(x^{k},y^{3})dx^{i},\
\mathbf{e}^{4}=dy^{4}+\partial _{i}n(x^{k})dx^{i},
\end{eqnarray*}
for $\ ^{v}\Upsilon =\Upsilon =\ ^{0}\Upsilon =const,$ where the
coefficients are subjected to conditions (\ref{sol1a})--(\ref{lccondm}) and
\begin{equation*}
\partial _{k}\omega +(\partial _{i}\phi /\phi ^{\ast })\omega ^{\ast
}-n_{i}\omega ^{\circ }=0,
\end{equation*}
defines solutions of the Einstein equations $R_{\alpha \beta }=(\
^{0}\Upsilon -4s^{2})g_{\alpha \beta }$ with nonholonomic interactions anf $%
f $--modifications of YMH fields encoded effectively into the vacuum
structure of GR with nontrivial cosmological constant, $\ ^{0}\Upsilon
-4s^{2}\neq 0.$
\end{corollary}

\begin{proof}
We have to consider $\ ^{v}\Upsilon =\Upsilon =\ ^{0}\Upsilon =const$ in the
Theorems \ref{th2b} and Corollary \ref{th3a}. $\square $
\end{proof}

\vskip5pt

Solutions of type (\ref{class2}) can be generated for conformal factors
being solutions of
\begin{equation*}
\partial _{k}\omega -\underline{w}_{i}(x^{k})\omega ^{\ast }+(\partial _{i}%
\underline{\phi }/\underline{\phi }^{\circ })\omega ^{\circ }=0
\end{equation*}
with respective \textquotedblright dual\textquotedblright\ generating
functions $\omega $ and $\underline{\phi }$ when the data (\ref{sol1a})--(%
\ref{lccondm}) $\ $are re--defined for solutions with weak Killing symmetry
on $\partial /\partial y^{3}.$

\subsubsection{Effective vacuum off--diagonal solutions for $f$-modifications%
}

Vacuum Einstein spaces encoding nonholonomic interactions of MGYMH fields
can be constructed using

\begin{corollary}
\label{corym2}An ansatz of type (\ref{class1}) with d--metric
\begin{eqnarray*}
\mathbf{g} &=&\epsilon _{i}e^{\psi (x^{k})}dx^{i}\otimes dx^{i}+\omega
^{2}(x^{k},y^{a})[(\ ^{0}h)^{2}\ \left[ \varpi ^{\ast }\left(
x^{i},y^{3}\right) \right] ^{2}\mathbf{e}^{3}\otimes \mathbf{e}^{3} \\
&&-\varpi ^{2}\left( x^{i},y^{3}\right) \underline{h}_{4}(x^{k},y^{4})%
\mathbf{e}^{4}\otimes \mathbf{e}^{4}], \\
\mathbf{e}^{3} &=&dy^{3}+w_{i}(x^{k},y^{3})dx^{i},\ \mathbf{e}%
^{4}=dy^{4}+n_{i}(x^{k})dx^{i},
\end{eqnarray*}%
where the coefficients are subjected to conditions (\ref{rel1})--(\ref{cond1}%
), (\ref{lccondm}) and
\begin{equation*}
\partial _{k}\omega -w_{i}\omega ^{\ast }-n_{i}\omega ^{\circ }=0,
\end{equation*}
define generic off--diagonal solutions of $R_{\alpha \beta }=0$.
\end{corollary}

\begin{proof}
It is a consequence of Theorem \ref{th2b} and Corollary \ref{corolvacuum}. $%
\square $
\end{proof}

\vskip5pt

Metric of class (\ref{class2}) are generated if the conformal factor is a
solution of%
\begin{equation*}
\partial _{k}\omega -\underline{w}_{i}(x^{k})\omega ^{\ast }-\underline{n}%
_{i}(x^{k},y^{4})\omega ^{\circ }=0
\end{equation*}
with respective \textquotedblright dual\textquotedblright\ generating
functions $\omega (x^{k},y^{a})$ and $\underline{\phi }(x^{k},y^{4});$ the
data and conditions (\ref{rel1})--(\ref{cond1}) and (\ref{lccondm}) are
reconsidered for ansatz with weak Killing symmetry on $\partial /\partial
y^{3}.$

\section{ $f$--modified and YMH Deformations of Black Holes}

\label{s4} In modified gravity, possible gauge--Higgs nonholonomic
interactions define off--diagonal deformations (for instance, of rotoid
type) of Schwarzschild black holes. In this section, we study effective EYMH
configurations when\ $\ ^{v}\Upsilon =\Upsilon $ and $\Upsilon +\
^{s}\lambda =0.$ Nonholonomic deformations can be derived from any
\textquotedblright prime\textquotedblright\ data $\left( \ ^{\circ }\mathbf{%
g,\ ^{\circ }A}_{\mu },\mathbf{\ ^{\circ }}\Phi \right) $ stating, for
instance, a diagonal cosmological monopole and non--Abbelian black hole
configuration in \cite{br3}. We can chose such a constant $s$ for $\
^{s}\lambda $ when the effective source is zero (if $\ ^{s}\lambda <0,$ this
is possible for $\Upsilon >0).$ The resulting nonholonomic matter field
configurations $\mathbf{\ A}_{\mu }=\mathbf{\ ^{\circ }A}_{\mu }+\ ^{\eta }%
\mathbf{A}_{\mu }$ (\ref{ans2a}), $F_{\mu \nu }=s\sqrt{|\mathbf{g}|}%
\varepsilon _{\mu \nu }$ (\ref{gaugstr1}) and $\mathbf{\ }\Phi =\ ^{\Phi
}\eta \mathbf{\ ^{\circ }}\Phi $ subjected to the conditions (\ref{cond3})
are encoded as vacuum off--diagonal polarizations into solutions of
equations (\ref{eq1})--(\ref{eq2}).

\subsection{(Non) holonomic and $f$--modified non--Abbelian effective vacuum
spaces}

We have to construct off--diagonal solutions of the Einstein equations for
the canonical d--connection taking the vacuum equations $\widehat{\mathbf{R}}%
_{\alpha \beta }=0$ and ansatz $~\mathbf{g}$ (\ref{dm1}) with coefficients
satisfying the conditions%
\begin{eqnarray}
&&\epsilon _{1}\psi ^{\bullet \bullet }(r,\theta )+\epsilon _{2}\psi
^{^{\prime \prime }}(r,\theta )=0;  \label{vacdsolc} \\
h_{3} &=&\pm e^{-2\ ^{0}\phi }\frac{\left( h_{4}^{\ast }\right) ^{2}}{h_{4}}%
\mbox{ for a given }h_{4}(r,\theta ,\varphi ),\ \phi (r,\theta ,\varphi )=\
^{0}\phi =const;\   \nonumber  \\
w_{i} &=&w_{i}(r,\theta ,\varphi ),\mbox{ for any such functions if }\lambda
=0;  \nonumber  \\
n_{i} &=&\ \left\{
\begin{array}{rcl}
\ \ \ ^{1}n_{i}(r,\theta )+\ ^{2}n_{i}(r,\theta )\int \left( h_{4}^{\ast
}\right) ^{2}|h_{4}|^{-5/2}dv,\ \  & \mbox{ if \ } & n_{i}^{\ast }\neq 0; \\
\ ^{1}n_{i}(r,\theta ),\quad \qquad \qquad \qquad \qquad \qquad &
\mbox{ if
\ } & n_{i}^{\ast }=0,%
\end{array}%
\right.  \nonumber
\end{eqnarray}%
when $h_{4}$ and $w_{i}$ are considered as generating functions. In general,
such effective vacuum solutions can be not generated in limits $\Upsilon +\
^{s}\lambda \rightarrow 0$ because of singularity of coefficients, for \
instance, for a class of solutions (\ref{sol1}) with coefficients (\ref%
{sol1a})--(\ref{sol1c}).

Imposing additional constraints on coefficients of d--metric, for $e^{-2\
^{0}\phi }=1,$ as solutions of (\ref{lccondm}),
\begin{eqnarray}
h_{3} &=&\pm 4\left[ \left( \sqrt{|h_{4}|}\right) ^{\ast }\right] ^{2},\quad
h_{4}^{\ast }\neq 0;  \label{auxvacsol} \\
w_{1}w_{2}\left( \ln |\frac{w_{1}}{w_{2}}|\right) ^{\ast } &=&w_{2}^{\bullet
}-w_{1}^{\prime },\ w_{i}^{\ast }\neq 0;\ w_{2}^{\bullet }-w_{1}^{\prime
}=0,\ w_{i}^{\ast }=0;  \nonumber  \\
\ ^{1}n_{1}^{\prime }(r,\theta )-\ ^{1}n_{2}^{\bullet }(r,\theta ) &=&0,\
n_{i}^{\ast }=0,  \label{vaclcsoc}
\end{eqnarray}%
we generate effective vacuum solutions of the Einstein equations for the
Levi--Civita connection.

The constructed class of vacuum solutions with coefficients subjected to
conditions (\ref{vacdsolc})--(\ref{vaclcsoc}) is of type (\ref{sol2}) for (%
\ref{ep1a})--(\ref{ep3a}). Such metrics consist a particular case of vacuum
ansatz defined by Corollary \ref{corym2} with $\underline{h}_{4}=1$ and $%
\omega =1.$

\subsection{Modifications of the Schwarzschild metric}

Let us consider a "prime" metric
\begin{equation}
~^{\varepsilon }\mathbf{g}=-d\xi \otimes d\xi -r^{2}(\xi )\ d\vartheta
\otimes d\vartheta -r^{2}(\xi )\sin ^{2}\vartheta \ d\varphi \otimes
d\varphi +\varkappa ^{2}(\xi )\ dt\otimes \ dt.  \label{5aux1}
\end{equation}%
In general, it is not obligatory to consider modifications only of solutions
of Einstein equations. Our goal is to construct a class of nonholonomic
deformations into "target" off--diagonal ones generating solutions of some
(effective) vacuum Einstein equations. The "primary" geometric data for (\ref%
{5aux1}) are stated by nontrivial coefficients
\begin{equation}
\check{g}_{1}=-1,\ \check{g}_{2}=-r^{2}(\xi ),\ \check{h}_{3}=-r^{2}(\xi
)\sin ^{2}\vartheta ,\ \check{h}_{4}=\varkappa ^{2}(\xi ),  \label{5aux1p}
\end{equation}%
for local coordinates $x^{1}=\xi ,x^{2}=\vartheta ,y^{3}=\varphi ,y^{4}=t,$
where
\begin{equation*}
\xi =\int dr\ \left\vert 1-\frac{2\mu _{0}}{r}+\frac{\varepsilon }{r^{2}}%
\right\vert ^{1/2}\mbox{ and }\varkappa ^{2}(r)=1-\frac{2\mu _{0}}{r}+\frac{%
\varepsilon }{r^{2}}.
\end{equation*}%
In a particular cas for $\varepsilon =0$ and $\mu _{0}$ considered as a
point mass, the metric $~^{\varepsilon }\mathbf{g}$ (\ref{5aux1}) determines
the Schwarzschild solution.

We generate exact solutions of the system (\ref{ep1a})--(\ref{ep3a}) with
effective $\ ^{v}\Upsilon =\Upsilon $ and $\Upsilon +\ ^{s}\lambda =0$ via
nonholonomic deformations $\ ^{\varepsilon }\mathbf{g\rightarrow }\ _{\eta
}^{\varepsilon }\mathbf{g,}$ when $g_{i}=\eta _{i}\check{g}_{i}$ and $%
h_{a}=\eta _{a}\check{h}_{a}$ and $w_{i},n_{i}.$ The resulting class of
target metrics is parameterized in the form {\small
\begin{eqnarray}
~_{\eta }^{\varepsilon }\mathbf{g} &=&\eta _{1}(\xi )d\xi \otimes d\xi +\eta
_{2}(\xi )r^{2}(\xi )\ d\vartheta \otimes d\vartheta +  \label{5sol1} \\
&&\eta _{3}(\xi ,\vartheta ,\varphi )r^{2}(\xi )\sin ^{2}\vartheta \ \delta
\varphi \otimes \delta \varphi -\eta _{4}(\xi ,\vartheta ,\varphi )\varpi
^{2}(\xi )\ \delta t\otimes \delta t,  \nonumber  \\
\delta \varphi &=&d\varphi +w_{1}(\xi ,\vartheta ,\varphi )d\xi +w_{2}(\xi
,\vartheta ,\varphi )d\vartheta ,\ \delta t=dt+n_{1}(\xi ,\vartheta )d\xi
+n_{2}(\xi ,\vartheta )d\vartheta ,  \nonumber
\end{eqnarray}%
} when the \ modified gravitational field equations for zero effective
source relate the prime and target coefficients of the vertical metric and
polarization functions via formulas
\begin{equation}
h_{3}=h_{0}^{2}(\varpi ^{\ast })^{2}=\eta _{3}(\xi ,\vartheta ,\varphi
)r^{2}(\xi )\sin ^{2}\vartheta ,\ h_{4}=-\varpi ^{2}=-\eta _{4}(\xi
,\vartheta ,\varphi )\varkappa ^{2}(\xi ).  \label{aux41}
\end{equation}
In these formulas, $|\eta _{3}|=(h_{0})^{2}|\check{h}_{4}/\check{h}_{3}|[(%
\sqrt{|\eta _{4}|})^{\ast }]^{2}$ and we have to chose $h_{0}=const$\ ($%
h_{0}=2$ in order to satisfy the first condition (\ref{vaclcsoc})). The
values $\check{h}_{a}$ are taken for the Schwarzschild solution for the
chosen system of coordinates and $\eta _{4}$ can be any function with $\eta
_{4}^{\ast }\neq 0.$ The $f$--modified gravitational polarizations $\eta
_{1} $ and $\eta _{2},$ when $\eta _{1}=\eta _{2}r^{2}=e^{\psi (\xi
,\vartheta )}, $ \ are found from (\ref{eq1}) with zero source, written in
the form $\psi ^{\bullet \bullet }+\psi ^{\prime \prime }=0.$

Introducing the coefficients (\ref{aux41}) in the ansatz (\ref{5sol1}), we
find a class of exact off--diagonal effective vacuum solutions of the
Einstein equations defining stationary nonholonomic deformations of the
Sch\-warz\-schild metric, {\small
\begin{eqnarray}
~^{\varepsilon }\mathbf{g} &=&-e^{\psi }\left( d\xi \otimes d\xi +\
d\vartheta \otimes d\vartheta \right) -4\left[ (\sqrt{|\eta _{4}|})^{\ast }%
\right] ^{2}\varkappa ^{2}\ \delta \varphi \otimes \ \delta \varphi +\eta
_{4}\varkappa ^{2}\ \delta t\otimes \delta t,  \nonumber  \\
\delta \varphi &=&d\varphi +w_{1}d\xi +w_{2}d\vartheta ,\ \delta t=dt+\
^{1}n_{1}d\xi +\ ^{1}n_{2}d\vartheta .  \label{5sol1a}
\end{eqnarray}%
} The N--connection coefficients $w_{i}(\xi ,\vartheta ,\varphi )$ and $\
^{1}n_{i}(\xi ,\vartheta )$ must satisfy the conditions (\ref{vaclcsoc}) in
order to get effect vacuum metrics with generic off--diagonal terms in GR.
Finally, we emphasize here that, in general, the bulk of solutions from the
set of target metrics do not define black holes and do not describe obvious
physical situations. $f$--modifications may preserve the singular character
of the coefficient $\varpi ^{2}$ vanishing on the horizon of a Schwarzschild
black hole if we take only smooth integration functions for some small
deformation parameters $\varepsilon .$

\subsection{Linear parametric polarizations and f--modifications induced by
YMH fields}

Let us select effective gravitational vacuum configurations with spherical
and/or rotoid (ellipsoid) symmetry if it is considered a generating function
\begin{equation}
\varpi ^{2}=\iota (\xi ,\vartheta ,\varphi )+\varepsilon \varrho (\xi
,\vartheta ,\varphi ).  \label{gf1}
\end{equation}%
For simplicity, we shall restrict our construction only to linear
decompositions on a small parameter $\varepsilon ,$ with $0<\varepsilon <<1.$

Using (\ref{gf1}), we compute $\left( \varpi ^{\ast }\right) ^{2}=[(\sqrt{%
|\iota |})^{\ast }]^{2}\ [1+\varepsilon \frac{1}{(\sqrt{|\iota |})^{\ast }}({%
\varrho }/\sqrt{|\iota |})^{\ast }]$ and the vertical coefficients of
d--metric (\ref{5sol1a}), i.e $h_{3}$ and $h_{4}$ (and corresponding
polarizations $\eta _{3}$ and $\eta _{4}),$ see formulas (\ref{aux41}). \
For rotoid configurations, \
\begin{equation}
\iota =1-\frac{2\mu (\xi ,\vartheta ,\varphi )}{r}\mbox{ and }\varrho =\frac{%
\iota _{0}(r)}{4\mu ^{2}}\sin (\omega _{0}\varphi +\varphi _{0}),
\label{aux42}
\end{equation}%
for $\mu (\xi ,\vartheta ,\varphi )=\mu _{0}+\varepsilon \mu _{1}(\xi
,\vartheta ,\varphi )$ (supposing that the mass is locally anisotropically
polarized) with certain constants $\mu ,\omega _{0}$ and $\varphi _{0}$ and
arbitrary functions/ polarizations $\mu _{1}(\xi ,\vartheta ,\varphi )$ and $%
\iota _{0}(r)$ to be determined from some boundary conditions, with $%
\varepsilon $ being the eccentricity. We may treat $\varepsilon $ as an
eccentricity imposing the condition that the coefficient $h_{4}=\varpi
^{2}=\eta _{4}(\xi ,\vartheta ,\varphi )\varkappa ^{2}(\xi )$ becomes zero
for data (\ref{aux42}) if
\begin{equation*}
r_{+}\simeq 2\mu _{0}/ (1+\varepsilon \frac{\iota _{0}(r)}{4\mu ^{2}}\sin
(\omega _{0}\varphi +\varphi _{0})).
\end{equation*}
\ Such conditions result in small deformations of the Schwarzschild
spherical horizon into an ellipsoidal one (rotoid configuration with
eccentricity $\varepsilon ).$

The resulting target solutions are for off--diagonal solution with rotoid
type symmetry
\begin{eqnarray}
~^{rot}\mathbf{g} &=&-e^{\psi }\left( d\xi \otimes d\xi +\ d\vartheta
\otimes d\vartheta \right) +\left( q+\varepsilon \varrho \right) \ \delta
t\otimes \delta t  \nonumber  \\
&&-4\left[ (\sqrt{|\iota |})^{\ast }\right] ^{2}\ [1+\varepsilon \frac{1}{(%
\sqrt{|\iota |})^{\ast }}({\varrho }/\sqrt{|\iota |})^{\ast }]\ \delta
\varphi \otimes \ \delta \varphi ,  \label{rotoidm} \\
\delta \varphi &=&d\varphi +w_{1}d\xi +w_{2}d\vartheta ,\ \delta t=dt+\
^{1}n_{1}d\xi +\ ^{1}n_{2}d\vartheta .  \nonumber
\end{eqnarray}%
The functions $\iota (\xi ,\vartheta ,\varphi )$ and $\varrho (\xi
,\vartheta ,\varphi )$ from (\ref{aux42}) and the N--connection coefficients
$w_{i}(\xi ,\vartheta ,\varphi )$ and $\ n_{i}=$ $\ ^{1}n_{i}(\xi ,\vartheta
)$ should be to conditions of type (\ref{vaclcsoc}),
\begin{eqnarray}
w_{1}w_{2}\left( \ln |\frac{w_{1}}{w_{2}}|\right) ^{\ast } &=&w_{2}^{\bullet
}-w_{1}^{\prime },\quad w_{i}^{\ast }\neq 0;  \label{constr3} \\
\mbox{ or \ }w_{2}^{\bullet }-w_{1}^{\prime } &=&0,\quad w_{i}^{\ast }=0;\
^{1}n_{1}^{\prime }(\xi ,\vartheta )-\ ^{1}n_{2}^{\bullet }(\xi ,\vartheta
)=0  \nonumber
\end{eqnarray}%
and $\psi (\xi ,\vartheta )$ being any function for which $\psi ^{\bullet
\bullet }+\psi ^{\prime \prime }=0,$ if we are interested to generate
Levi--Civita configurations.

Off--diagonal rotoid deformations of black hole solutions in GR are possible
via $f$--deformations, in noncommutative gravity, by nonlinear YMH
interactions and via generic off--diagonal Einstein gravitational fields.
The generating functions and parameters of such solutions depend on the type
of gravity model we consider.

\section{Ellipsoid--Solitonic f--modifications of EYMH Configurations}

\label{s5} It is possible to prescribe nonholonomic constraints with $\
^{v}\Upsilon =\Upsilon =\ ^{0}\Upsilon =const$ and $\ ^{0}\Upsilon +\
^{s}\lambda \neq 0.$ This allows us to construct off--diagonal solutions for
MGYMH systems (\ref{eq1})--(\ref{eq4}) and (\ref{lccondm}) with coefficients
of metric of type (\ref{sol1}). Such metrics provide explicit examples of
effective non--vacuum solutions with ansatz for metrics considered for the
Corollary \ref{corym1} with $\underline{h}_{4}=1$ and $\omega =1.$

\subsection{Nonholonomic rotoid $f$--modifications}

Using the anholonomic frame method, we can generate a class of solutions
with nontrivial cosmological constant possessing different limits (for large
radial distances and small nonholonomic deformations) than the vacuum
configurations considered in previous section.

Let us consider a diagonal metric of type
\begin{equation}
~_{\lambda }^{\varepsilon }\mathbf{g}=d\xi \otimes d\xi +r^{2}(\xi )\
d\theta \otimes d\theta +r^{2}(\xi )\sin ^{2}\theta \ d\varphi \otimes
d\varphi +\ _{\lambda }\varkappa ^{2}(\xi )\ dt\otimes \ dt,  \label{sds1}
\end{equation}%
where nontrivial metric coefficients are parametriz\-ed in the form $\check{g%
}_{1}=1,\ \check{g}_{2}=r^{2}(\xi ),\ \check{h}_{3}=r^{2}(\xi )\sin
^{2}\vartheta ,\ \check{h}_{4}=\ _{\lambda }\varkappa ^{2}(\xi )$, for local
coordinates $x^{1}=\xi ,x^{2}=\vartheta ,y^{3}=\varphi ,y^{4}=t,$ with $\xi
=\int $ $dr/\left\vert q(r)\right\vert ^{\frac{1}{2}},$ and $\ _{\lambda
}\varkappa ^{2}(r)=-\sigma ^{2}(r)q(r),$ for $q(r)=1-2m(r)/r-\Lambda
r^{2}/3. $ In variables $\left( r,\theta ,\varphi \right) ,$ the metric (\ref%
{sds1}) is equivalent to (\ref{ansatz1}).

The ansatz for such classes of solutions is chosen in the form {\small
\begin{eqnarray*}
\ ^{\lambda }\mathbf{\mathring{g}} &=&e^{\underline{\phi }(\xi ,\theta )}\
(d\xi \otimes d\xi +\ d\theta \otimes d\theta )+h_{3}(\xi ,\theta ,\varphi
)\ {\delta }\varphi \otimes {\delta }\varphi +h_{4}(\xi ,\theta ,\varphi )\ {%
\delta t}\otimes ~{\delta t}, \\
~\delta \varphi &=&d\varphi +w_{1}\left( \xi ,\theta ,\varphi \right) d\xi
+w_{2}\left( \xi ,\theta ,\varphi \right) d\theta , \\
\delta t &=&dt+n_{1}\left( \xi ,\theta ,\varphi \right) d\xi +n_{2}\left(
\xi ,\theta ,\varphi \right) d\theta ,
\end{eqnarray*}%
} for $h_{3}=-h_{0}^{2}(\varpi ^{\ast })^{2}=\eta _{3}(\xi ,\theta ,\varphi
)r^{2}(\xi )\sin ^{2}\vartheta ,\ h_{4}=b^{2}=\eta _{4}(\xi ,\theta ,\varphi
)\ _{\lambda }\varkappa ^{2}(\xi ).$ The coefficients \ of this metric
determine exact solutions if
\begin{eqnarray}
&&\underline{\phi }^{\bullet \bullet }(\xi ,\theta )+\underline{\phi }%
^{^{\prime \prime }}(\xi ,\theta )=2(\ ^{0}\Upsilon +\ ^{s}\lambda );
\label{anhsol2} \\
h_{3} &=&\pm \frac{\left( \phi ^{\ast }\right) ^{2}}{4\ (\ ^{0}\Upsilon +\
^{s}\lambda )}e^{-2\ ^{0}\phi (\xi ,\theta )},\ h_{4}=\mp \frac{1}{4\ (\
^{0}\Upsilon +\ ^{s}\lambda )}e^{2(\phi -\ ^{0}\phi (\xi ,\theta ))};  \nonumber
\\
w_{i} &=&\partial _{i}\phi /\phi ^{\ast };  \nonumber  \\
n_{i} &=&\ ^{1}n_{i}(\xi ,\theta )+\ ^{2}n_{i}(\xi ,\theta )\int \left( \phi
^{\ast }\right) ^{2}e^{-2(\phi -\ ^{0}\phi (\xi ,\vartheta ))}d\varphi ,\
\nonumber  \\
&=&\ \ \left\{
\begin{array}{rcl}
\ \ ^{1}n_{i}(\xi ,\theta )+\ ^{2}n_{i}(\xi ,\theta )\int e^{-4\phi }\frac{%
\left( h_{4}^{\ast }\right) ^{2}}{h_{4}}d\varphi ,\ \  & \mbox{ if \ } &
n_{i}^{\ast }\neq 0; \\
\ ^{1}n_{i}(\xi ,\theta ),\quad \qquad \qquad \qquad \qquad \qquad &
\mbox{
if \ } & n_{i}^{\ast }=0;%
\end{array}%
\right.  \nonumber
\end{eqnarray}%
for any nonzero coefficients $h_{a}$ and $h_{a}^{\ast }$ and arbitrary
integrating functions, $^{1}n_{i}(\xi ,\theta ),\ ^{2}n_{i}(\xi ,\theta ),$
and generating functions, $\phi (\xi ,\theta ,\varphi )$ and $\ ^{0}\phi
(\xi ,\theta ).$ Such values have to be determined from certain boundary
conditions for a fixed system of coordinates and following additional
assumptions depending on the type of $f$--modified theory of gravity we
study.

For nonholonomic ellipsoid de Sitter configurations, we parameterize
\begin{eqnarray}
~_{\lambda }^{rot}\mathbf{g} &=&-e^{\underline{\phi }(\xi ,\theta )}\left(
d\xi \otimes d\xi +\ d\theta \otimes d\theta \right) +\left( \underline{%
\iota }+\varepsilon \underline{\varrho }\right) \ \delta t\otimes \delta t
\nonumber  \\
&&-h_{0}^{2}\left[ (\sqrt{|\underline{\iota }|})^{\ast }\right]
^{2}[1+\varepsilon \frac{1}{(\sqrt{|\underline{\iota }|})^{\ast }}(%
\underline{\varrho }/\sqrt{|\underline{\iota }|})^{\ast }]\ \delta \varphi
\otimes \ \delta \varphi ,  \nonumber  \\
\delta \varphi &=&d\varphi +w_{1}d\xi +w_{2}d\vartheta ,\ \delta
t=dt+n_{1}d\xi +n_{2}d\vartheta ,  \label{soladel}
\end{eqnarray}%
where $\underline{\iota }=1-\frac{2\ ^{1}\underline{\mu }(r,\theta ,\varphi )%
}{r}$ and $\ \ \underline{\varrho }=\frac{\underline{\iota }_{0}(r)}{4%
\underline{\mu }_{0}^{2}}\sin (\omega _{0}\varphi +\varphi _{0})$ are fixed
for anisotro\-pic rotoid configurations on the "smaller
horizon\textquotedblright\ (when $\ h_{4}=0),$%
\begin{equation*}
\ r_{+}\simeq 2\ ^{1}\underline{\mu }/ ( 1+\varepsilon \frac{\underline{%
\iota }_{0}(r)}{4\underline{\mu }_{0}^{2}}\sin (\omega _{0}\varphi +\varphi
_{0})),
\end{equation*}
for a corresponding $\underline{\iota }_{0}(r).$

For the Levi--Civita configurations, we have to consider additional
nonholonomic constraints resulting in zero torsion in order to generate
solutions of the Einstein equations for the Levi--Civita connection.
Following the condition (\ref{constr3}), for $\phi ^{\ast }\neq 0,$ we
obtain that $\phi (r,\varphi ,\theta )=\ln |h_{4}^{\ast }/\sqrt{|h_{3}h_{4}|}%
|$ must be any function defined in non--explicit form from equation $%
2e^{2\phi }\phi =\ ^{0}\Upsilon +\ ^{s}\lambda .$ It is possible to solve
the set of constraints for the N--connection coefficients the integration
functions in (\ref{anhsol2}) are subjected to $w_{1}w_{2}\left( \ln |\frac{%
w_{1}}{w_{2}}|\right) ^{\ast }=w_{2}^{\bullet }-w_{1}^{\prime }$ for $%
w_{i}^{\ast }\neq 0;$ $w_{2}^{\bullet }-w_{1}^{\prime }=0$ for$\ w_{i}^{\ast
}=0;$ and take $\ n_{i}=\ ^{1}n_{i}(x^{k})$ for $\ ^{1}n_{1}^{\prime
}(x^{k})-\ ^{1}n_{2}^{\bullet }(x^{k})=0.$

\subsection{Modifications via effective vacuum solitons}

Off--diagonal modifications in effective vacuum spacetimes can be modeled by
3--d solitonic gravitational interactions with nontrivial vertical conformal
factor $\omega .$ In this section, we suppose that there are satisfied the
conditions of Corollary \ref{corym2} with $\underline{h}_{4}=1$ for
effective vacuum solutions. Such prime and target metrics may encode MGYMH
configurations and their nonlinear wave deformations. Additional constraints
for the Levi--Civita configurations may result in EYMH solutions.

\subsubsection{Solitonic waves for the conformal factor $\protect\omega %
(x^{1},y^{3},t)$}

We consider functions $\omega =\eta (x^{1},y^{3},t),$ when $y^{4}=t$ is a
time like coordinate, determined by a solution of KdP equation \cite{kadom},
\begin{equation}
\pm \eta ^{\ast \ast }+(\partial _{t}\eta +\eta \ \eta ^{\bullet }+\epsilon
\eta ^{\bullet \bullet \bullet })^{\bullet }=0,  \label{kdp1}
\end{equation}%
with dispersion $\epsilon .$ In the dispersionless limit $\epsilon
\rightarrow 0$ the solutions are independent on $y^{3}$ and transform into
those given by Burgers' equation $\partial _{t}\eta +\eta \ \eta ^{\bullet
}=0.$ The conditions (\ref{lccondm}) are written in the form $\mathbf{e}%
_{1}\eta =\eta ^{\bullet }+w_{1}(x^{i},y^{3})\eta ^{\ast
}+n_{1}(x^{i})\partial _{t}\eta =0$. For $\eta ^{\prime }=0,$ we can impose
the condition $w_{2}=0$ and $n_{2}=0.$

The corresponding effective vacuum solitonic $f$--modifications are given by
\begin{eqnarray*}
\ _{1}\mathbf{g} &=&e^{\psi (x^{k})}(dx^{1}\otimes dx^{1}+dx^{2}\otimes
dx^{2})+\left[ \eta (x^{1},y^{3},t)\right] ^{2}h_{a}(x^{1},y^{3})\ \mathbf{e}%
^{a}\otimes \mathbf{e}^{a}, \\
\mathbf{e}^{3} &=&dy^{3}+w_{1}(x^{k},y^{3})dx^{1},\ \mathbf{e}%
^{4}=dy^{4}+n_{1}(x^{k})dx^{1}.
\end{eqnarray*}%
This class of metrics depend on all spacetime coordinates and may not
posses, in general, Killing symmetries. Nevertheless, there are symmetries
determined by solitonic solutions of (\ref{kdp1}). Alternatively, we can
consider\ that $\eta $ is a solution of any three dimensional solitonic and/
or other nonlinear wave equations; in a similar manner, we can generate
solutions for $\omega =\eta (x^{2},y^{3},t).$

\subsubsection{$f$--modifications with solitonic factor $\protect\omega %
(x^{i},t)$}

There are off--diagonal solutions when the effective vacuum metrics are with
a solitonic dynamics not depending on anisotropic coordinate $y^{3}.$ To
generate such nonholonomic configurations we take $\omega =\widehat{\eta }%
(x^{k},t)$ is a solution of KdP equation
\begin{equation}
\pm \widehat{\eta }^{\bullet \bullet }+(\partial _{t}\widehat{\eta }+%
\widehat{\eta }\ \widehat{\eta }^{\prime }+\epsilon \widehat{\eta }^{\prime
\prime \prime })^{\prime }=0  \label{kdp3a}
\end{equation}%
and consider that in the dispersionless limit $\epsilon \rightarrow 0$ the
solutions are independent on $x^{1}$ and determined by Burgers' equation $%
\partial _{t}\widehat{\eta }+\widehat{\eta }\ \widehat{\eta }^{\prime }=0.$

A class of effective vacuum solitonic EYMH configurations encoding $f$%
--modifications from MGYMH interactions is given by
\begin{eqnarray*}
\ _{2}\mathbf{g} &=&e^{\psi (x^{k})}(dx^{1}\otimes dx^{1}+dx^{2}\otimes
dx^{2})+\left[ \widehat{\eta }(x^{k},t)\right] ^{2}h_{a}(x^{k},y^{3})\
\mathbf{e}^{a}\otimes \mathbf{e}^{a}, \\
\mathbf{e}^{3} &=&dy^{3}+w_{1}(x^{k},y^{3})dx^{1},\ \mathbf{e}%
^{4}=dy^{4}+n_{1}(x^{k})dx^{1},
\end{eqnarray*}%
when (\ref{lccondm}) are equivalent to $\mathbf{e}_{1}\widehat{\eta }=%
\widehat{\eta }^{\bullet }+n_{1}(x^{i})\partial _{t}\widehat{\eta }=0,\
\mathbf{e}_{2}\widehat{\eta }=\widehat{\eta }^{\prime }+n_{2}(x^{i})\partial
_{t}\widehat{\eta }=0$.

Finally, we note that modified gravity theories can be characterized by
exact solutions with an infinite number of vacuum gravitational 2-d and 3-d
configurations stated by corresponding solitonic hierarchies and
bi--Hamilton structures, for instance, related to different KdP equations (%
\ref{kdp3a}). There are possible mixtures with solutions for 2-d and 3-d
sine--Gordon equations etc, see details in Ref. \cite{vacarsolitonhier}. The
constants, parametric dependence and generating functions are determined by
corresponding models of modified gravity and possible extra dimension
generalizations.

\section{Concluding Remarks}

As a consequence of the discovery of the accelerating expansion of the
Universe and attempts to formulate self--consistent schemes and propose an
experimentally verifiable phenomenology for quantum gravity a number of
modified gravity (MG) theories have been proposed along recent years. It is
considered that a change of the paradigm of standard particle theory is
inevitable in order to understand and solve the dark energy and dark matter
problems. In this sense, the $f(R,T,...)$--theories with functional
dependence on various types of scalar curvatures, torsions, energy--momentum
tensors etc have become popular candidates which may be capable to solve
various puzzles in particle physics and modern cosmology.

It is considered that viable modified gravity theories should be
characterized by a well behavior at local scales when cosmological effects
like inflation and late--time acceleration are reproduced. For any candidate
model to a modified/generalized gravity theory, to construct exact solutions
with physical importance, describing nonlinear gravitational and matter
field interactions, is a technically difficult task which requests new
sophisticate geometric, analytic and numerical methods. Such exact solutions
present an important theoretical tool for understanding properties of
gravity theories at the classical level and suggest a number of ideas how a
quantum formalism has to be developed in order to include possible
modifications and corrections to cosmological and related microscopic
scenarios.

In this work, we have shown that the $f$--modified gravitational field
equations and generalizations with Yang--Mills and Higgs equations (in
brief, MGYMH) can be solved in very general forms using the so--called
anholonomic frame deformation method, AFDM. The approach was elaborated in a
series of works on geometric methods of constructing exact solutions in
Einstein gravity and its (noncommutative) generalized Finsler, brane, string
modifications, see reviews of results in Refs. \cite%
{ijgmmp,vexsol1,vexsol2,veymh}. One of the most important features of the
AFDM is that it propose a set of geometric constructions for decoupling
certain physically important systems of nonlinear partial differential
equations, PD, with respect to certain classes of nonholonomic frames. More
than that, the method shows how can integrate such PDE in general form, with
generic off--diagonal metrics depending on all spacetime coordinates via
various integration and generating functions, symmetry parameters etc.

One of the most important conclusion of our work is that using "auxiliary"
connections necessary for decoupling PDEs and generating off--diagonal
configurations for metrics we can mimic various classes of $f$%
--modifications. In many cases, the geometric constructions and solutions
can be constrained to be interpreted in the framework of the general
relativity (GR) theory. For instance, we shown how black hole solutions in
GR may be deformed (with small parameter, or in certain general forms) into
off--diagonal metrics if possible $f$--modifications and YMH interactions
are taken into considerations. New classes of metrics and connections may be
with certain Killing and/or solitonic symmetries or deformed into
non--Killing configurations. For well--defined conditions, a subclass of
such metrics can be generated to have nontrivial limits to effective vacuum
solutions, or with nonhomogeneous/ anisotropic polarizations of cosmological
constants and gravitational--matter interactions. This support a
conservative opinion that a number of modifications which seem to be
necessary in modern cosmology and for elaborating quantum gravity models can
be alternatively explained by off--diagonal, parametric and/or nonholonomic
interactions in GR.

Nevertheless, the AFDM was originally elaborated, and generalized, for
various modified theories of gravity. It allows us to prescribe, for
instance, a convenient value of the scalar curvature for an auxiliary
connection (such a curvature is not fixed, in general, for the Levi--Civita
connection) or certain type of generalized matter sources and locally
anisotropic nonlinear polarizations of interaction constants. Choosing a
convenient nonholonomic $2+2+2+....$ -- splitting, we can generate
off--diagonal solutions in four and extra dimensions with cosmologically
observable anisotropic behavior and related, for instance, to effective
renormalized theories, see reviews of results in \cite%
{odints2,vgrg,covquant,vaxiom}). The main result of this paper is that we
provided explicit proofs and explicit examples that the MGYMH equations and
possible effective EYMH systems can be solved and studied using "pure"
geometric and analytic methods. To quantize such nonlinear classical
modified gravitational and matter field systems and study possible
implications in QCD physics \cite{vel1,vel2} is a plan for our future work.

\vskip5pt

\textbf{Acknowledgments:\ } The SV work is partially supported by the
Program IDEI, PN-II-ID-PCE-2011-3-0256 and performed for a visit supported
by the physics department at Kocaely University (Ismit, Turkey). \appendix

\setcounter{equation}{0} \renewcommand{\theequation}
{A.\arabic{equation}} \setcounter{subsection}{0}
\renewcommand{\thesubsection}
{A.\arabic{subsection}}

\section{2+2 Splitting of Lorentz Manifolds}

\label{sa}We provide the main results and formulas on the canonical
d--connection and corresponding d--torsion and d--curvature.

\begin{theorem}
\label{thadist} In coefficient form, the distortion relations (\ref{distrel}%
) are computed
\begin{equation}
\Gamma _{\ \alpha \beta }^{\gamma }=\widehat{\mathbf{\Gamma }}_{\ \alpha
\beta }^{\gamma }+\widehat{\mathbf{Z}}_{\ \alpha \beta }^{\gamma },
\label{distrel1}
\end{equation}%
where the distortion tensor $\widehat{\mathbf{Z}}_{\ \alpha \beta }^{\gamma
} $ is {\small
\begin{eqnarray}
\ Z_{jk}^{i} &=&Z_{bc}^{a}=0,\ Z_{jk}^{a}=-\widehat{C}_{jb}^{i}g_{ik}g^{ab}-%
\frac{1}{2}\Omega _{jk}^{a},~Z_{bk}^{i}=\frac{1}{2}\Omega
_{jk}^{c}g_{cb}g^{ji}-\Xi _{jk}^{ih}~\widehat{C}_{hb}^{j},  \nonumber  \\
Z_{bk}^{a} &=&\ ^{+}\Xi _{cd}^{ab}~\widehat{T}_{kb}^{c},\ Z_{kb}^{i}=\frac{1%
}{2}\Omega _{jk}^{a}g_{cb}g^{ji}+\Xi _{jk}^{ih}~\widehat{C}_{hb}^{j},
\label{deft} \\
Z_{jb}^{a} &=&\ ^{-}\Xi _{cb}^{ad}~\widehat{T}_{jd}^{c},\ Z_{ab}^{i}=-\frac{%
g^{ij}}{2}\left[ \widehat{T}_{ja}^{c}g_{cb}+\widehat{T}_{jb}^{c}g_{ca}\right]
,\   \nonumber
\end{eqnarray}%
} for $\ \Xi _{jk}^{ih}=\frac{1}{2}(\delta _{j}^{i}\delta
_{k}^{h}-g_{jk}g^{ih})$ and $~^{\pm }\Xi _{cd}^{ab}=\frac{1}{2}(\delta
_{c}^{a}\delta _{d}^{b}+g_{cd}g^{ab}).$ The nontrivial coefficients $\Omega
_{jk}^{a}$ and $\ \widehat{\mathbf{T}}_{\ \alpha \beta }^{\gamma }$ are
given, respectively, by formulas (\ref{anhcoef}) and, see below, (\ref{dtors}%
).
\end{theorem}

\begin{proof}
It follows from a straightforward verification in N--adapted frames (\ref%
{nader}) and (\ref{nadif}) that the sums of coefficients (\ref{cdc}) and (%
\ref{deft}) result in the coefficients of the Levi--Civita connection $%
\Gamma _{\ \alpha \beta }^{\gamma }$ for a general metric parameterized as a
d--metric $\mathbf{g}=[g_{ij},g_{ab}]$ \ (\ref{dm}). $\square $
\end{proof}

\vskip5pt

\begin{theorem}
\label{thadtors}The nonholonomically induced torsion $\widehat{\mathcal{T}}$
$=\{\widehat{\mathbf{T}}_{\ \alpha \beta }^{\gamma }\}$ of $\ \widehat{%
\mathbf{D}}$ is determined in a unique form by the metric compatibility
condition, $\widehat{\mathbf{D}}\mathbf{g}=0,$ and zero horizontal and
vertical torsion coefficients, $\widehat{T}_{\ jk}^{i}=0$ and $\widehat{T}%
_{\ bc}^{a}=0,$ but with nontrivial h--v-- coefficients {\small
\begin{equation}
\widehat{T}_{\ jk}^{i}=\widehat{L}_{jk}^{i}-\widehat{L}_{kj}^{i},\widehat{T}%
_{\ ja}^{i}=\widehat{C}_{jb}^{i},\widehat{T}_{\ ji}^{a}=-\Omega _{\ ji}^{a},%
\widehat{T}_{aj}^{c}=\widehat{L}_{aj}^{c}-e_{a}(N_{j}^{c}),\widehat{T}_{\
bc}^{a}=\ \widehat{C}_{bc}^{a}-\ \widehat{C}_{cb}^{a}.  \label{dtors}
\end{equation}%
}
\end{theorem}

\begin{proof}
The coefficients (\ref{dtors}) are computed by introducing $\ D=$ $\widehat{%
\mathbf{D}},$ with coefficients (\ref{cdc}), and $X=\mathbf{e}_{\alpha },Y=%
\mathbf{e}_{\beta }$ (for N--adapted frames (\ref{nader})) into standard
formula for torsion, $\mathcal{T}(X,Y):=D_{\mathbf{X}}Y-D_{\mathbf{Y}%
}X-[X,Y] $.

$\square $
\end{proof}

\vskip5pt

In a similar form, introducing $\widehat{\mathbf{D}}$ and $X=\mathbf{e}%
_{\alpha },Y=\mathbf{e}_{\beta },Z=$ $\mathbf{e}_{\gamma }$ into $\mathcal{R}%
(X,Y):=D_{\mathbf{X}}D_{\mathbf{Y}}-D_{\mathbf{Y}}D_{\mathbf{X}}-D_{\mathbf{%
[X,Y]}},$ we prove

\begin{theorem}
\label{thasa1}The curvature $\widehat{\mathcal{R}}=\{\widehat{\mathbf{R}}_{\
\beta \gamma \delta }^{\alpha }\}$ of $\widehat{\mathbf{D}}$ is
characterized by N--adapted coefficients {\small
\begin{eqnarray}
\widehat{R}_{\ hjk}^{i} &=&e_{k}\widehat{L}_{\ hj}^{i}-e_{j}\widehat{L}_{\
hk}^{i}+\widehat{L}_{\ hj}^{m}\widehat{L}_{\ mk}^{i}-\widehat{L}_{\ hk}^{m}%
\widehat{L}_{\ mj}^{i}-\widehat{C}_{\ ha}^{i}\Omega _{\ kj}^{a},  \nonumber  \\
\widehat{R}_{\ bjk}^{a} &=&e_{k}\widehat{L}_{\ bj}^{a}-e_{j}\widehat{L}_{\
bk}^{a}+\widehat{L}_{\ bj}^{c}\widehat{L}_{\ ck}^{a}-\widehat{L}_{\ bk}^{c}%
\widehat{L}_{\ cj}^{a}-\widehat{C}_{\ bc}^{a}\Omega _{\ kj}^{c},
\label{dcurv} \\
\widehat{R}_{\ jka}^{i} &=&e_{a}\widehat{L}_{\ jk}^{i}-\widehat{D}_{k}%
\widehat{C}_{\ ja}^{i}+\widehat{C}_{\ jb}^{i}\widehat{T}_{\ ka}^{b},\widehat{%
R}_{\ bka}^{c}=e_{a}\widehat{L}_{\ bk}^{c}-D_{k}\widehat{C}_{\ ba}^{c}+%
\widehat{C}_{\ bd}^{c}\widehat{T}_{\ ka}^{c},  \nonumber  \\
\widehat{R}_{\ jbc}^{i} &=&e_{c}\widehat{C}_{\ jb}^{i}-e_{b}\widehat{C}_{\
jc}^{i}+\widehat{C}_{\ jb}^{h}\widehat{C}_{\ hc}^{i}-\widehat{C}_{\ jc}^{h}%
\widehat{C}_{\ hb}^{i},  \nonumber  \\
\widehat{R}_{\ bcd}^{a} &=&e_{d}\widehat{C}_{\ bc}^{a}-e_{c}\widehat{C}_{\
bd}^{a}+\widehat{C}_{\ bc}^{e}\widehat{C}_{\ ed}^{a}-\widehat{C}_{\ bd}^{e}%
\widehat{C}_{\ ec}^{a}.  \nonumber
\end{eqnarray}%
}
\end{theorem}

We can re--define the differential geometry of a (pseudo) Riemannian space $%
\mathbf{V}$ in nonholonomic form in terms of geometric data $(\mathbf{g,}%
\widehat{\mathbf{D}})$ which is equivalent to the ''standard'' formulation
with $(\mathbf{g,\nabla }).$

\begin{corollary}
\label{corolricci} The Ricci tensor $\widehat{\mathbf{R}}_{\alpha \beta }:=%
\widehat{\mathbf{R}}_{\ \alpha \beta \gamma }^{\gamma }$ (\ref{riccid}) of $%
\ \widehat{\mathbf{D}}$ is characterized by N--adapted coefficients{\small
\begin{equation}
\widehat{\mathbf{R}}_{\alpha \beta }=\{\widehat{R}_{ij}:=\widehat{R}_{\
ijk}^{k},\ \widehat{R}_{ia}:=-\widehat{R}_{\ ika}^{k},\ \widehat{R}_{ai}:=%
\widehat{R}_{\ aib}^{b},\ \widehat{R}_{ab}:=\widehat{R}_{\ abc}^{c}\}.
\label{driccic}
\end{equation}%
}
\end{corollary}

\begin{proof}
The formulas for $h$--$v$--components (\ref{driccic}) are obtained by
contracting respectively the coefficients (\ref{dcurv}). Using $\widehat{%
\mathbf{D}}$ (\ref{cdc}), we express such formulas in terms of partial
derivatives of coefficients of metric $\mathbf{g}$ (\ref{mst}) and any
equivalent parametrization in the form (\ref{dm}), or (\ref{ansatz}). $%
\square $
\end{proof}

\vskip5pt

The scalar curvature $\ ^{s}\widehat{R}$ of $\ \widehat{\mathbf{D}}$ is by
definition
\begin{equation}
\ ^{s}\widehat{R}:=\mathbf{g}^{\alpha \beta }\widehat{\mathbf{R}}_{\alpha
\beta }=g^{ij}\widehat{R}_{ij}+g^{ab}\widehat{R}_{ab}.  \label{sdcurv}
\end{equation}

Using (\ref{driccic}) and (\ref{sdcurv}), we can compute the Einstein tensor
$\widehat{\mathbf{E}}_{\alpha \beta }$ of $\widehat{\mathbf{D}},$
\begin{equation}
\widehat{\mathbf{E}}_{\alpha \beta }:= \widehat{\mathbf{R}}_{\alpha \beta }-%
\frac{1}{2}\mathbf{g}_{\alpha \beta }\ ^{s}\widehat{R}.  \label{enstdt}
\end{equation}%
In general, this tensor is different from that constructed using (\ref%
{riccie}) for the Levi--Civita connection $\nabla .$

\begin{proposition}
\label{approp}The N--adapted coefficients $\widehat{\mathbf{\Gamma }}_{\
\alpha \beta }^{\gamma }$ of $\ \widehat{\mathbf{D}}$ are identic to the
coefficients $\Gamma _{\ \alpha \beta }^{\gamma }$ of $\nabla $, both sets
computed with respect to N--adapted frames (\ref{nader}) and (\ref{nadif}),
if and only if there are satisfied the conditions $\widehat{L}%
_{aj}^{c}=e_{a}(N_{j}^{c}),\widehat{C}_{jb}^{i}=0$ and $\Omega _{\
ji}^{a}=0. $
\end{proposition}

\begin{proof}
If the conditions of the Proposition, i.e. constraints (\ref{lcconstr}), are
satisfied, all N--adapted coefficients of the torsion $\widehat{\mathbf{T}}%
_{\ \alpha \beta }^{\gamma }$ \ (\ref{dtors}) are zero. In such a case, the
distortion tensor $\widehat{\mathbf{Z}}_{\ \alpha \beta }^{\gamma }$ is also
zero. Following formula (\ref{distrel1}), we get $\Gamma _{\ \alpha \beta
}^{\gamma }=\widehat{\mathbf{\Gamma }}_{\ \alpha \beta }^{\gamma }.$
Inversely, if the last equalities of coefficients are satisfied for a chosen
splitting (\ref{ncon}), we get trivial torsions and distortions of $\nabla .$
We emphasize that, in general, $\widehat{\mathbf{D}}\neq \nabla $ because
such connections have different transformation rules under frame/coordinate
transforms. Nevertheless, \ if $\Gamma _{\ \alpha \beta }^{\gamma }=\widehat{%
\mathbf{\Gamma }}_{\ \alpha \beta }^{\gamma }$ in a N--adapted frame of
reference, we get corresponding equalities for the Riemann and Ricci tensors
etc. This means that the N--coefficients are such way fixed via frame
transforms that the nonholonomic distribution became integrable even, in
general, the frames (\ref{nader}) and (\ref{nadif}) are nonholonomic
(because not all anholonomy coefficients are not obligatory zero, for
instance, $w_{ia}^{b}=\partial _{a}N_{i}^{b}$ may be nontrivial, see
formulas (\ref{anhcoef})). $\square $
\end{proof}

\vskip5pt

In order to elaborate models of gravity theories for $\nabla $ and/or $%
\widehat{\mathbf{D}}$, we have to consider the corresponding Ricci tensors,%
\begin{eqnarray}
Ric=\{R_{\ \beta \gamma } &:=&R_{\ \beta \gamma \alpha }^{\alpha }\},%
\mbox{ for
}\nabla =\{\Gamma _{\ \alpha \beta }^{\gamma }\},  \label{riccie} \\
\mbox{ and }\widehat{R}ic &=&\{\widehat{\mathbf{R}}_{\ \beta \gamma }:=%
\widehat{\mathbf{R}}_{\ \beta \gamma \alpha }^{\alpha }\},\mbox{ for }%
\widehat{\mathbf{D}}=\{\widehat{\mathbf{\Gamma }}_{\ \alpha \beta }^{\gamma
}\}.  \label{riccid}
\end{eqnarray}

\setcounter{equation}{0} \renewcommand{\theequation}
{B.\arabic{equation}} \setcounter{subsection}{0}
\renewcommand{\thesubsection}
{B.\arabic{subsection}}

\section{Proof of Theorem \protect\ref{th2a}}

\label{sb}For $\omega =1$ and $\underline{h}_{a}=const,$ such proofs can be
obtained by straightforward computations \cite{ijgmmp}, see also Appendices
to \cite{veymh}. The approach was extended for $\omega \neq 1$ and higher
dimensions in \cite{vexsol1,vexsol2}. In this section, we sketch a proof for
ansatz (\ref{ans1}) with nontrivial $\underline{h}_{4}$ depending on
variable $y^{4}$ when $\omega =1$ in data (\ref{paramdcoef}). At the next
step, the formulas will be completed for nontrivial values $\omega \neq 1$.

Using $\widehat{R}_{~1}^{1}=\widehat{R}_{~2}^{2}$ and $\widehat{R}_{~3}^{3}=%
\widehat{R}_{~4}^{4},$ the equations (\ref{cdeinst}) for $\widehat{\mathbf{D}%
}$ and data (\ref{data1a}) (see below) can be written for any source (\ref%
{source}) in the form {\small
\begin{equation*}
\widehat{E}_{~1}^{1}=\widehat{E}_{~2}^{2}=-\widehat{R}_{~3}^{3}=\mathbf{%
\Upsilon }(x^{k},y^{3})+\underline{\mathbf{\Upsilon }}(x^{k},y^{3},y^{4}),\
\widehat{E}_{~3}^{3}=\widehat{E}_{~4}^{4}=-\widehat{R}_{~1}^{1}=\ ^{v}%
\mathbf{\Upsilon }(x^{k}).
\end{equation*}%
} The geometric data for the conditions of Theorem \ref{th2a} are $%
g_{i}=g_{i}(x^{k})$ and {\small
\begin{equation}
g_{3}=h_{3}(x^{k},y^{3}),g_{4}=h_{4}(x^{k},y^{3})\underline{h}%
_{4}(x^{k},y^{4}),N_{i}^{3}=w_{i}(x^{k},y^{3}),N_{i}^{4}=n_{i}(x^{k},y^{3}),
\label{data1a}
\end{equation}%
} for $\underline{h}_{3}=1$ and local coordinates $u^{\alpha
}=(x^{i},y^{a})=(x^{1},x^{2},y^{3},y^{4}).$ \ For such values, we shall
compute respectively the coefficients of $\Omega _{\ \alpha \beta }^{\ a}$
in (\ref{anhcoef}), canonical d--connection $\widehat{\mathbf{\Gamma }}_{\
\alpha \beta }^{\gamma }$ (\ref{cdc}), d--torsion $\widehat{\mathbf{T}}_{\
\alpha \beta }^{\gamma }$\ (\ref{dtors}), necessary coefficients of
d--curvature $\widehat{\mathbf{R}}_{\ \alpha \beta \gamma }^{\tau }$ (\ref%
{dcurv}) with respective contractions for $\widehat{\mathbf{R}}_{\alpha
\beta }:=\widehat{\mathbf{R}}_{\ \alpha \beta \gamma }^{\gamma }$ (\ref%
{driccic}) and resulting $\ ^{s}\widehat{R}$ (\ref{sdcurv}) and $\widehat{%
\mathbf{E}}_{\alpha \beta }$(\ref{enstdt}). Finally, we shall state the
conditions (\ref{lcconstr}) when general coefficients (\ref{data1a}) are
considered for d--metrics.

\subsection{The coefficients of the canonical d--connection and its torsion}

There are such nontrivial coefficients of $\widehat{\mathbf{\Gamma }}_{\
\alpha \beta }^{\gamma }$ (\ref{cdc}), {\small
\begin{eqnarray}
\widehat{L}_{11}^{1} &=&\frac{g_{1}^{\bullet }}{2g_{1}},\ \widehat{L}%
_{12}^{1}=\frac{g_{1}^{\prime }}{2g_{1}},\widehat{L}_{22}^{1}=-\frac{%
g_{2}^{\bullet }}{2g_{1}},\ \widehat{L}_{11}^{2}=\frac{-g_{1}^{\prime }}{%
2g_{2}},\ \widehat{L}_{12}^{2}=\frac{g_{2}^{\bullet }}{2g_{2}},\ \widehat{L}%
_{22}^{2}=\frac{g_{2}^{\prime }}{2g_{2}},  \label{nontrdc} \\
\widehat{L}_{4k}^{4} &=&\frac{\mathbf{\partial }_{k}(h_{4}\underline{h}_{4})%
}{2h_{4}\underline{h}_{4}}-\frac{w_{k}h_{4}^{\ast }}{2h_{4}}-(n_{k}+%
\underline{n}_{k})\frac{\underline{h}_{4}^{\circ }}{2\underline{h}_{4}},%
\widehat{L}_{3k}^{3}=\frac{\mathbf{\partial }_{k}h_{3}}{2h_{3}}-\frac{%
w_{k}h_{3}^{\ast }}{2h_{3}},\widehat{L}_{4k}^{3}=\frac{h_{4}\underline{h}_{4}%
}{-2h_{3}}n_{k}^{\ast },  \nonumber  \\
\widehat{L}_{3k}^{4} &=&\frac{1}{2}n_{k}^{\ast },\widehat{C}_{33}^{3}=\frac{%
h_{3}^{\ast }}{2h_{3}},\widehat{C}_{44}^{3}=-\frac{h_{4}^{\ast }\underline{h}%
_{4}}{h_{3}},\ \widehat{C}_{33}^{4}=-\frac{h_{3}\underline{h}_{3}^{\circ }}{%
h_{4}\underline{h}_{4}},~\widehat{C}_{34}^{4}=\frac{h_{4}^{\ast }}{2h_{4}},%
\widehat{C}_{44}^{4}=\frac{\underline{h}_{4}^{\circ }}{2\underline{h}_{4}}.
\nonumber
\end{eqnarray}%
} We shall need also the values
\begin{equation}
\ \widehat{C}_{3}=\widehat{C}_{33}^{3}+\widehat{C}_{34}^{4}=\frac{%
h_{3}^{\ast }}{2h_{3}}+\frac{h_{4}^{\ast }}{2h_{4}},\widehat{C}_{4}=\widehat{%
C}_{43}^{3}+\widehat{C}_{44}^{4}=\frac{\underline{h}_{4}^{\circ }}{2%
\underline{h}_{4}}.  \label{aux3}
\end{equation}

Using data (\ref{data1a}) for $\underline{w}_{i}=\underline{n}_{i}=0,$ the
coefficients $\Omega _{ij}^{a}=\mathbf{e}_{j}\left( N_{i}^{a}\right) -%
\mathbf{e}_{i}(N_{j}^{a})$ (\ref{anhcoef}), are computed
\begin{equation*}
\Omega _{ij}^{a}=\mathbf{\partial }_{j}\left( N_{i}^{a}\right) -\partial
_{i}(N_{j}^{a})-w_{i}(N_{j}^{a})^{\ast }+w_{j}(N_{i}^{a})^{\ast }.
\end{equation*}%
There are such nontrivial values {\small
\begin{eqnarray}
\Omega _{12}^{3} &=&-\Omega _{21}^{3}=\mathbf{\partial }_{2}w_{1}-\partial
_{1}w_{2}-w_{1}w_{2}^{\ast }+w_{2}w_{1}^{\ast }=w_{1}^{\prime
}-w_{2}^{\bullet }-w_{1}w_{2}{}^{\ast }+w_{2}w_{1}^{\ast }{};  \nonumber  \\
\Omega _{12}^{4} &=&-\Omega _{21}^{4}=\mathbf{\partial }_{2}n_{1}-\partial
_{1}n_{2}-w_{1}n_{2}^{\ast }+w_{2}n_{1}^{\ast }=n_{1}^{\prime
}-n_{2}^{\bullet }-w_{1}n_{2}^{\ast }{}+w_{2}n_{1}^{\ast }{}.  \label{omeg}
\end{eqnarray}%
} The nontrivial coefficients of d--torsion (\ref{dtors}) are $\widehat{T}%
_{\ ji}^{a}=-\Omega _{\ ji}^{a}$ (\ref{omeg}) and $\widehat{T}_{aj}^{c}=%
\widehat{L}_{aj}^{c}-e_{a}(N_{j}^{c}).$ For other types of coefficients,
{\small
\begin{eqnarray}
\widehat{T}_{\ jk}^{i} &=&\widehat{L}_{jk}^{i}-\widehat{L}_{kj}^{i}=0,~%
\widehat{T}_{\ ja}^{i}=\widehat{C}_{jb}^{i}=0,~\widehat{T}_{\ bc}^{a}=\
\widehat{C}_{bc}^{a}-\ \widehat{C}_{cb}^{a}=0,  \nonumber  \\
\widehat{T}_{3k}^{3} &=&\widehat{L}_{3k}^{3}-e_{3}(N_{k}^{3})=\frac{\mathbf{%
\partial }_{k}h_{3}}{2h_{3}}-w_{k}\frac{h_{3}^{\ast }}{2h_{3}}-w_{k}^{\ast
}{},  \nonumber  \\
\widehat{T}_{4k}^{3} &=&\widehat{L}_{4k}^{3}-e_{4}(N_{k}^{3})=-\frac{h_{4}%
\underline{h}_{4}}{2h_{3}}n_{k}^{\ast },\ \widehat{T}_{3k}^{4}=~\widehat{L}%
_{3k}^{4}-e_{3}(N_{k}^{4})=\frac{1}{2}n_{k}^{\ast }-n_{k}^{\ast }=-\frac{1}{2%
}n_{k}^{\ast },  \nonumber  \\
\widehat{T}_{4k}^{4} &=&\widehat{L}_{4k}^{4}-e_{4}(N_{k}^{4})=\frac{\mathbf{%
\partial }_{k}(h_{4}\underline{h}_{4})}{2h_{4}\underline{h}_{4}}-w_{k}\frac{%
h_{4}^{\ast }}{2h_{4}}-n_{k}\frac{\underline{h}_{4}^{\circ }}{2\underline{h}%
_{4}},  \nonumber  \\
-\widehat{T}_{12}^{3} &=&w_{1}^{\prime }-w_{2}^{\bullet }-w_{1}w_{2}^{\ast
}{}+w_{2}w_{1}^{\ast },\ -\widehat{T}_{12}^{4}=n_{1}^{\prime
}-n_{2}^{\bullet }-w_{1}n_{2}^{\ast }{}+w_{2}n_{1}^{\ast }{}.
\label{nontrtors}
\end{eqnarray}%
} Such coefficients of torsion if and only if $\Gamma _{\ \alpha \beta
}^{\gamma }=\widehat{\mathbf{\Gamma }}_{\ \alpha \beta }^{\gamma }.$

\subsection{The zero torsion conditions}

We must solve the equations
\begin{equation*}
\widehat{T}_{4k}^{4}=\widehat{L}_{4k}^{4}-e_{4}(N_{k}^{4})=\frac{\mathbf{%
\partial }_{k}(h_{4}\underline{h}_{4})}{2h_{4}\underline{h}_{4}}-w_{k}\frac{%
h_{4}^{\ast }}{2h_{4}}-n_{k}\frac{\underline{h}_{4}^{\circ }}{2\underline{h}%
_{4}}=0,
\end{equation*}%
which follow from formulas (\ref{nontrtors}). Taking any $\underline{h}_{4}$
for which
\begin{equation}
n_{k}\underline{h}_{4}^{\circ }=\mathbf{\partial }_{k}\underline{h}_{4},
\label{aux4}
\end{equation}%
\ the condition $n_{k}\frac{h_{4}^{\ast }}{2h_{4}}\frac{\underline{h}%
_{4}^{\circ }}{2\underline{h}_{4}}-\frac{h_{4}^{\ast }}{2h_{4}}\frac{\mathbf{%
\partial }_{k}\underline{h}_{4}}{2\underline{h}_{4}}=0$ can be satisfied.
For instance, parametrizing $\underline{h}_{4}=~^{h}\underline{h}_{4}(x^{k})%
\underline{h}(y^{4}),$ the equations (\ref{aux4}) are solved by any
\begin{equation*}
\underline{h}(y^{4})=e^{\varkappa y^{4}}\mbox{ and }n_{k}=\varkappa \partial
_{k}[~^{h}\underline{h}_{4}(x^{k})],\mbox{ for }\varkappa =const.
\end{equation*}

We conclude that for any $n_{k}$ and $\underline{h}_{4}$ related by
conditions (\ref{aux4}) the zero torsion conditions (\ref{nontrtors}) are
the same as for $\underline{h}_{4}=const.$ Using a similar proof from \cite%
{vexsol1,vexsol2}, it is possible to verify by straightforward computations
that $\widehat{T}_{\beta \gamma }^{\alpha }=0$ if the equations (\ref{lccond}%
) are solved.

\subsection{N--adapted coefficients of the canonical Ricci d--tensor}

The values $\widehat{R}_{ij}=\widehat{R}_{\ ijk}^{k}$ are computed as (\ref%
{riccid}) using (\ref{dcurv}),
\begin{eqnarray*}
\widehat{R}_{\ hjk}^{i} &=&\mathbf{e}_{k}\widehat{L}_{.hj}^{i}-\mathbf{e}_{j}%
\widehat{L}_{hk}^{i}+\widehat{L}_{hj}^{m}\widehat{L}_{mk}^{i}-\widehat{L}%
_{hk}^{m}\widehat{L}_{mj}^{i}-\widehat{C}_{ha}^{i}\Omega _{jk}^{a} \\
&=&\mathbf{\partial }_{k}\widehat{L}_{.hj}^{i}-\partial _{j}\widehat{L}%
_{hk}^{i}+\widehat{L}_{hj}^{m}\widehat{L}_{mk}^{i}-\widehat{L}_{hk}^{m}%
\widehat{L}_{mj}^{i}.
\end{eqnarray*}%
We note $\widehat{C}_{\ ha}^{i}=0$ and {\small
\begin{equation*}
\mathbf{e}_{k}\widehat{L}_{hj}^{i}=\partial _{k}\widehat{L}%
_{hj}^{i}+N_{k}^{a}\partial _{a}\widehat{L}_{hj}^{i}=\partial _{k}\widehat{L}%
_{hj}^{i}+w_{k}\left( \widehat{L}_{hj}^{i}\right) ^{\ast }+n_{k}\left(
\widehat{L}_{hj}^{i}\right) ^{\circ }=\partial _{k}\widehat{L}_{hj}^{i}.
\end{equation*}%
} Taking derivatives of (\ref{nontrdc}), we obtain {\small
\begin{eqnarray*}
\partial _{1}\widehat{L}_{\ 11}^{1} &=&(\frac{g_{1}^{\bullet }}{2g_{1}}%
)^{\bullet }=\frac{g_{1}^{\bullet \bullet }}{2g_{1}}-\frac{\left(
g_{1}^{\bullet }\right) ^{2}}{2\left( g_{1}\right) ^{2}},\ \partial _{1}%
\widehat{L}_{\ 12}^{1}=(\frac{g_{1}^{\prime }}{2g_{1}})^{\bullet }=\frac{%
g_{1}^{\prime \bullet }}{2g_{1}}-\frac{g_{1}^{\bullet }g_{1}^{\prime }}{%
2\left( g_{1}\right) ^{2}},\  \\
\partial _{1}\widehat{L}_{\ 22}^{1} &=&(-\frac{g_{2}^{\bullet }}{2g_{1}}%
)^{\bullet }=-\frac{g_{2}^{\bullet \bullet }}{2g_{1}}+\frac{g_{1}^{\bullet
}g_{2}^{\bullet }}{2\left( g_{1}\right) ^{2}},\ \partial _{1}\widehat{L}_{\
11}^{2}=(-\frac{g_{1}^{\prime }}{2g_{2}})^{\bullet }=-\frac{g_{1}^{\prime
\bullet }}{2g_{2}}+\frac{g_{1}^{\bullet }g_{2}^{\prime }}{2\left(
g_{2}\right) ^{2}}, \\
\partial _{1}\widehat{L}_{\ 12}^{2} &=&(\frac{g_{2}^{\bullet }}{2g_{2}}%
)^{\bullet }=\frac{g_{2}^{\bullet \bullet }}{2g_{2}}-\frac{\left(
g_{2}^{\bullet }\right) ^{2}}{2\left( g_{2}\right) ^{2}},\ \partial _{1}%
\widehat{L}_{\ 22}^{2}=(\frac{g_{2}^{\prime }}{2g_{2}})^{\bullet }=\frac{%
g_{2}^{\prime \bullet }}{2g_{2}}-\frac{g_{2}^{\bullet }g_{2}^{\prime }}{%
2\left( g_{2}\right) ^{2}},
\end{eqnarray*}%
\begin{eqnarray*}
\partial _{2}\widehat{L}_{\ 11}^{1} &=&(\frac{g_{1}^{\bullet }}{2g_{1}}%
)^{\prime }=\frac{g_{1}^{\bullet \prime }}{2g_{1}}-\frac{g_{1}^{\bullet
}g_{1}^{\prime }}{2\left( g_{1}\right) ^{2}},~\partial _{2}\widehat{L}_{\
12}^{1}=(\frac{g_{1}^{\prime }}{2g_{1}})^{\prime }=\frac{g_{1}^{\prime
\prime }}{2g_{1}}-\frac{\left( g_{1}^{\prime }\right) ^{2}}{2\left(
g_{1}\right) ^{2}}, \\
\partial _{2}\widehat{L}_{\ 22}^{1} &=&(-\frac{g_{2}^{\bullet }}{2g_{1}}%
)^{\prime }=-\frac{g_{2}^{\bullet ^{\prime }}}{2g_{1}}+\frac{g_{2}^{\bullet
}g_{1}^{^{\prime }}}{2\left( g_{1}\right) ^{2}},\ \partial _{2}\widehat{L}%
_{\ 11}^{2}=(-\frac{g_{1}^{\prime }}{2g_{2}})^{\prime }=-\frac{g_{1}^{\prime
\prime }}{2g_{2}}+\frac{g_{1}^{\bullet }g_{1}^{\prime }}{2\left(
g_{2}\right) ^{2}}, \\
\partial _{2}\widehat{L}_{\ 12}^{2} &=&(\frac{g_{2}^{\bullet }}{2g_{2}}%
)^{\prime }=\frac{g_{2}^{\bullet \prime }}{2g_{2}}-\frac{g_{2}^{\bullet
}g_{2}^{\prime }}{2\left( g_{2}\right) ^{2}},\partial _{2}\widehat{L}_{\
22}^{2}=(\frac{g_{2}^{\prime }}{2g_{2}})^{\prime }=\frac{g_{2}^{\prime
\prime }}{2g_{2}}-\frac{\left( g_{2}^{\prime }\right) ^{2}}{2\left(
g_{2}\right) ^{2}}.
\end{eqnarray*}%
} These values result in two nontrivial components,
\begin{eqnarray*}
\widehat{R}_{\ 212}^{1} &=&\frac{g_{2}^{\bullet \bullet }}{2g_{1}}-\frac{%
g_{1}^{\bullet }g_{2}^{\bullet }}{4\left( g_{1}\right) ^{2}}-\frac{\left(
g_{2}^{\bullet }\right) ^{2}}{4g_{1}g_{2}}+\frac{g_{1}^{\prime \prime }}{%
2g_{1}}-\frac{g_{1}^{\prime }g_{2}^{\prime }}{4g_{1}g_{2}}-\frac{\left(
g_{1}^{\prime }\right) ^{2}}{4\left( g_{1}\right) ^{2}}, \\
\widehat{R}_{\ 112}^{2} &=&-\frac{g_{2}^{\bullet \bullet }}{2g_{2}}+\frac{%
g_{1}^{\bullet }g_{2}^{\bullet }}{4g_{1}g_{2}}+\frac{\left( g_{2}^{\bullet
}\right) ^{2}}{4(g_{2})^{2}}-\frac{g_{1}^{\prime \prime }}{2g_{2}}+\frac{%
g_{1}^{\prime }g_{2}^{\prime }}{4(g_{2})^{2}}+\frac{\left( g_{1}^{\prime
}\right) ^{2}}{4g_{1}g_{2}}.
\end{eqnarray*}%
Considering $\widehat{R}_{11}=-\widehat{R}_{\ 112}^{2}$ and $\widehat{R}%
_{22}=\widehat{R}_{\ 212}^{1},$ for $g^{i}=1/g_{i},$ we compute
\begin{equation*}
\widehat{R}_{1}^{1}=\widehat{R}_{2}^{2}=-\frac{1}{2g_{1}g_{2}}%
[g_{2}^{\bullet \bullet }-\frac{g_{1}^{\bullet }g_{2}^{\bullet }}{2g_{1}}-%
\frac{\left( g_{2}^{\bullet }\right) ^{2}}{2g_{2}}+g_{1}^{\prime \prime }-%
\frac{g_{1}^{\prime }g_{2}^{\prime }}{2g_{2}}-\frac{(g_{1}^{\prime })^{2}}{%
2g_{1}}],
\end{equation*}%
which is contained in equations (\ref{eq2}).

To derive the equations (\ref{eq3}) we consider the third formula in (\ref%
{dcurv}),
\begin{eqnarray*}
\widehat{R}_{\ bka}^{c} &=&\frac{\partial \widehat{L}_{bk}^{c}}{\partial
y^{a}}-\widehat{C}_{~ba|k}^{c}+\widehat{C}_{~bd}^{c}\widehat{T}_{~ka}^{d} \\
&=&\frac{\partial \widehat{L}_{bk}^{c}}{\partial y^{a}}-(\frac{\partial
\widehat{C}_{ba}^{c}}{\partial x^{k}}+\widehat{L}_{dk}^{c\,}\widehat{C}%
_{ba}^{d}-\widehat{L}_{bk}^{d}\widehat{C}_{da}^{c}-\widehat{L}_{ak}^{d}%
\widehat{C}_{bd}^{c})+\widehat{C}_{bd}^{c}\widehat{T}_{ka}^{d}.
\end{eqnarray*}%
Contracting indices, we obtain $\widehat{R}_{bk}=\widehat{R}_{\ bka}^{a}=%
\frac{\partial L_{bk}^{a}}{\partial y^{a}}-\widehat{C}_{ba|k}^{a}+\widehat{C}%
_{bd}^{a}\widehat{T}_{ka}^{d},$ where for $\widehat{C}_{b}:=\widehat{C}%
_{ba}^{c}$%
\begin{eqnarray*}
\widehat{C}_{b|k} &=&\mathbf{e}_{k}\widehat{C}_{b}-\widehat{L}_{\ bk}^{d\,}%
\widehat{C}_{d}=\partial _{k}\widehat{C}_{b}-N_{k}^{e}\partial _{e}\widehat{C%
}_{b}-\widehat{L}_{\ bk}^{d\,}\widehat{C}_{d} \\
&=&\partial _{k}\widehat{C}_{b}-w_{k}\widehat{C}_{b}^{\ast }-n_{k}\widehat{C}%
_{b}^{\circ }-\widehat{L}_{\ bk}^{d\,}\widehat{C}_{d}.
\end{eqnarray*}%
We consider a conventional splitting $\widehat{R}_{bk}=\ _{[1]}R_{bk}+\
_{[2]}R_{bk}+\ _{[3]}R_{bk},$ where%
\begin{eqnarray*}
\ _{[1]}R_{bk} &=&\left( \widehat{L}_{bk}^{3}\right) ^{\ast }+\left(
\widehat{L}_{bk}^{4}\right) ^{\circ },\ _{[2]}R_{bk}=-\partial _{k}\widehat{C%
}_{b}+w_{k}\widehat{C}_{b}^{\ast }+n_{k}\widehat{C}_{b}^{\circ }+\widehat{L}%
_{\ bk}^{d\,}\widehat{C}_{d}, \\
\ _{[3]}R_{bk} &=&\widehat{C}_{bd}^{a}\widehat{T}_{ka}^{d}=\widehat{C}%
_{b3}^{3}\widehat{T}_{k3}^{3}+\widehat{C}_{b4}^{3}\widehat{T}_{k3}^{4}+%
\widehat{C}_{b3}^{4}\widehat{T}_{k4}^{3}+\widehat{C}_{b4}^{4}\widehat{T}%
_{k4}^{4}.
\end{eqnarray*}%
Using formulas (\ref{nontrdc}), (\ref{nontrtors}) and (\ref{aux3}), we
compute {\small
\begin{eqnarray*}
\ _{[1]}R_{3k} &=&\left( \widehat{L}_{3k}^{3}\right) ^{\ast }+\left(
\widehat{L}_{3k}^{4}\right) ^{\circ }=\left( \frac{\mathbf{\partial }%
_{k}h_{3}}{2h_{3}}-w_{k}\frac{h_{3}^{\ast }}{2h_{3}}\right) ^{\ast} \\
&=&-w_{k}^{\ast }\frac{h_{3}^{\ast }}{2h_{3}}-w_{k}\left( \frac{h_{3}^{\ast }%
}{2h_{3}}\right) ^{\ast }+\frac{1}{2}\left( \frac{\mathbf{\partial }_{k}h_{3}%
}{h_{3}}\right) ^{\ast }, \\
\ _{[2]}R_{3k} &=&-\partial _{k}\widehat{C}_{3}+w_{k}\widehat{C}_{3}^{\ast
}+n_{k}\widehat{C}_{3}^{\circ }+\widehat{L}_{\ 3k}^{3\,}\widehat{C}_{3}+%
\widehat{L}_{\ 3k}^{4\,}\widehat{C}_{4}= \\
&=&w_{k}[\frac{h_{3}^{\ast \ast }}{2h_{3}}-\frac{3}{4}\frac{(h_{3}^{\ast
})^{2}}{(h_{3})^{2}}+\frac{h_{4}^{\ast \ast }}{2h_{4}}-\frac{1}{2}\frac{%
(h_{4}^{\ast })^{2}}{(h_{4})^{2}}-\frac{1}{4}\frac{h_{3}^{\ast }}{h_{3}}%
\frac{h_{4}^{\ast }}{h_{4}}]+n_{k}^{\ast }\frac{\underline{h}_{4}^{\circ }}{4%
\underline{h}_{4}} \\
&&+\left( \frac{\mathbf{\partial }_{k}h_{3}}{2h_{3}}+\frac{\mathbf{\partial }%
_{k}\underline{h}_{3}}{2\underline{h}_{3}}\right) (\frac{h_{3}^{\ast }}{%
2h_{3}}+\frac{h_{4}^{\ast }}{2h_{4}})-\frac{1}{2}\partial _{k}(\frac{%
h_{3}^{\ast }}{h_{3}}+\frac{h_{4}^{\ast }}{h_{4}}), \\
\ _{[3]}R_{3k} &=&\widehat{C}_{33}^{3}\widehat{T}_{k3}^{3}+\widehat{C}%
_{34}^{3}\widehat{T}_{k3}^{4}+\widehat{C}_{33}^{4}\widehat{T}_{k4}^{3}+%
\widehat{C}_{34}^{4}\widehat{T}_{k4}^{4} = w_{k}\left( \frac{(h_{3}^{\ast
})^{2}}{4(h_{3})^{2}}+\frac{(h_{4}^{\ast })^{2}}{4(h_{4})^{2}}\right) \\
&&+ w_{k}^{\ast }\frac{h_{3}^{\ast }}{2h_{3}} + n_{k}\frac{h_{4}^{\ast }}{%
2h_{4}}\frac{\underline{h}_{4}^{\circ }}{2\underline{h}_{4}}-\frac{%
h_{3}^{\ast }}{2h_{3}}\frac{\mathbf{\partial }_{k}h_{3}}{2h_{3}}-\frac{%
h_{4}^{\ast }}{2h_{4}}\left( \frac{\mathbf{\partial }_{k}h_{4}}{2h_{4}}+%
\frac{\mathbf{\partial }_{k}\underline{h}_{4}}{2\underline{h}_{4}}\right) .
\end{eqnarray*}%
} Putting together, we get%
\begin{eqnarray}
\ \widehat{R}_{3k} &=&w_{k}\left[ \frac{h_{4}^{\ast \ast }}{2h_{4}}-\frac{1}{%
4}\frac{(h_{4}^{\ast })^{2}}{(h_{4})^{2}}-\frac{1}{4}\frac{h_{3}^{\ast }}{%
h_{3}}\frac{h_{4}^{\ast }}{h_{4}}\right] +n_{k}^{\ast }\frac{\underline{h}%
_{4}^{\circ }}{4\underline{h}_{4}}+n_{k}\frac{h_{4}^{\ast }}{2h_{4}}\frac{%
\underline{h}_{4}^{\circ }}{2\underline{h}_{4}}  \nonumber  \\
&&+\frac{h_{4}^{\ast }}{2h_{4}}\frac{\mathbf{\partial }_{k}h_{3}}{2h_{3}}-%
\frac{1}{2}\frac{\partial _{k}h_{4}^{\ast }}{h_{4}}+\frac{1}{4}\frac{%
h_{4}^{\ast }\partial _{k}h_{4}}{(h_{4})^{2}}-\frac{h_{4}^{\ast }}{2h_{4}}%
\frac{\mathbf{\partial }_{k}\underline{h}_{4}}{2\underline{h}_{4}},  \nonumber
\end{eqnarray}%
which is equivalent to (\ref{eq3}) if the conditions $n_{k}\underline{h}%
_{4}^{\circ }=\mathbf{\partial }_{k}\underline{h}_{4},$ see below formula (%
\ref{aux4}), are satisfied.

The values $\ \widehat{R}_{4k}=\ _{[1]}R_{4k}+\ _{[2]}R_{4k}+\ _{[3]}R_{4k},
$ are defined by
\begin{eqnarray*}
&&\ _{[1]}R_{4k}=\left( \widehat{L}_{4k}^{3}\right) ^{\ast }+\left( \widehat{%
L}_{4k}^{4}\right) ^{\circ },\ \ _{[2]}R_{4k}=-\partial _{k}\widehat{C}%
_{4}+w_{k}\widehat{C}_{4}^{\ast }+n_{k}\widehat{C}_{4}^{\circ }+\widehat{L}%
_{\ 4k}^{3\,}\widehat{C}_{3} \\
&&+\widehat{L}_{\ 4k}^{4\,}\widehat{C}_{4},\ \ _{[3]}R_{4k}=\widehat{C}%
_{4d}^{a}\widehat{T}_{ka}^{d}=\widehat{C}_{43}^{3}\widehat{T}_{k3}^{3}+%
\widehat{C}_{44}^{3}\widehat{T}_{k3}^{4}+\widehat{C}_{43}^{4}\widehat{T}%
_{k4}^{3}+\widehat{C}_{44}^{4}\widehat{T}_{k4}^{4}.
\end{eqnarray*}%
Using $\widehat{L}_{4k}^{3}$ and $\widehat{L}_{4k}^{4}$ from (\ref{nontrdc}%
), we obtain{\small
\begin{eqnarray*}
&&\ _{[1]}R_{4k}=\left( \widehat{L}_{4k}^{3}\right) ^{\ast }+\left( \widehat{%
L}_{4k}^{4}\right) ^{\circ }=-(\frac{h_{4}\underline{h}_{4}}{2h_{3}}%
n_{k}^{\ast }{})^{\ast }+(\frac{\mathbf{\partial }_{k}(h_{4}\underline{h}%
_{4})}{2h_{4}\underline{h}_{4}}-w_{k}\frac{h_{4}^{\ast }}{2h_{4}}-n_{k})%
\frac{\underline{h}_{4}^{\circ }}{2\underline{h}_{4}})^{\circ } \\
&=&-n_{k}^{\ast \ast }{}\frac{h_{4}}{2h_{3}}\underline{h}_{4}-n_{k}^{\ast }(%
\frac{h_{4}^{\ast }}{2h_{3}}-\frac{h_{4}^{\ast }h_{3}^{\ast }}{%
2(h_{3})^{\ast }})\underline{h}_{4}-n_{k}(\frac{\underline{h}_{4}^{\circ
\circ }}{2\underline{h}_{4}}-\frac{(\underline{h}_{4}^{\circ })^{2}}{2(%
\underline{h}_{4})^{2}})+\frac{\mathbf{\partial }_{k}\underline{h}%
_{4}^{\circ }}{2\underline{h}_{4}}-\frac{\underline{h}_{4}^{\circ }\mathbf{%
\partial }_{k}\underline{h}_{4}}{2(\underline{h}_{4})^{2}}.
\end{eqnarray*}%
} The second term follows from (\ref{aux3}), for $\widehat{C}_{3}, \widehat{C%
}_{4},$ and (\ref{nontrdc}), for $\ \widehat{L}_{4k}^{3}$ and $\widehat{L}%
_{4k}^{4},$ {\small
\begin{eqnarray*}
\ _{[2]}R_{4k} &=&-\partial _{k}\widehat{C}_{4}+w_{k}\widehat{C}_{4}^{\ast
}+n_{k}\widehat{C}_{4}^{\circ }+\widehat{L}_{\ 4k}^{3\,}\widehat{C}_{3}+%
\widehat{L}_{\ 4k}^{4\,}\widehat{C}_{4} \\
&=&-w_{k}\left( \frac{h_{4}^{\ast }}{2h_{4}}\frac{\underline{h}_{4}^{\circ }%
}{2\underline{h}_{4}}\right) -n_{k}^{\ast }{}\frac{h_{4}\underline{h}_{4}}{%
2h_{3}\underline{h}_{3}}\left( \frac{h_{3}^{\ast }}{2h_{3}}+\frac{%
h_{4}^{\ast }}{2h_{4}}\right) \\
&&+n_{k}\left[ \left( \frac{\underline{h}_{4}^{\circ }}{2\underline{h}_{4}}%
\right) ^{\circ }-\frac{\underline{h}_{4}^{\circ }}{2\underline{h}_{4}}\frac{%
\underline{h}_{4}^{\circ }}{2\underline{h}_{4}}\right] +\frac{\mathbf{%
\partial }_{k}h_{4}}{2h_{4}}\frac{\underline{h}_{4}^{\circ }}{2\underline{h}%
_{4}}+\frac{\mathbf{\partial }_{k}\underline{h}_{4}}{2\underline{h}_{4}}%
\frac{\underline{h}_{4}^{\circ }}{2\underline{h}_{4}}-\frac{\partial _{k}%
\underline{h}_{4}^{\circ }}{2\underline{h}_{4}}+\frac{\underline{h}%
_{4}^{\circ }\partial _{k}\underline{h}_{4}}{2(\underline{h}_{4})^{2}}.
\end{eqnarray*}%
} The formulas (\ref{nontrdc}), with $\widehat{C}_{43}^{3},\widehat{C}%
_{44}^{3},\widehat{C}_{43}^{4},\widehat{C}_{44}^{4},$ and the formulas (\ref%
{nontrtors}), with $\widehat{T}_{k3}^{3},\widehat{T}_{k3}^{4},\widehat{T}%
_{k4}^{3},\widehat{T}_{k4}^{4},$ have to be used for the third term,%
\begin{eqnarray*}
_{\lbrack 3]}R_{4k} &=&\widehat{C}_{43}^{3}\widehat{T}_{k3}^{3}+\widehat{C}%
_{44}^{3}\widehat{T}_{k3}^{4}+\widehat{C}_{43}^{4}\widehat{T}_{k4}^{3}+%
\widehat{C}_{44}^{4}\widehat{T}_{k4}^{4} \\
&=&w_{k}\frac{h_{4}^{\ast }}{2h_{4}}\frac{\underline{h}_{4}^{\circ }}{2%
\underline{h}_{4}}+n_{k}(\frac{\underline{h}_{4}^{\circ }}{2\underline{h}_{4}%
})^{2}-\frac{\mathbf{\partial }_{k}h_{4}}{2h_{4}}\frac{\underline{h}%
_{4}^{\circ }}{2\underline{h}_{4}}-\frac{\underline{h}_{4}^{\circ }}{2%
\underline{h}_{4}}\frac{\mathbf{\partial }_{k}\underline{h}_{4}}{2\underline{%
h}_{4}}.
\end{eqnarray*}
Summarizing above three terms,{\small
\begin{eqnarray*}
&&\widehat{R}_{4k}=-n_{k}^{\ast \ast }{}\frac{h_{4}}{2h_{3}}\underline{h}%
_{4}+n_{k}^{\ast }{}\left( -\frac{h_{4}^{\ast }}{2h_{3}}+\frac{h_{4}^{\ast
}h_{3}^{\ast }}{2(h_{3})^{\ast }}-\frac{h_{4}^{\ast }h_{3}^{\ast }}{%
4(h_{3})^{\ast }}-\frac{h_{4}^{\ast }}{4h_{3}}\right) \underline{h}_{4} \\
&&+n_{k}\left( -\frac{\underline{h}_{4}^{\circ \circ }}{2\underline{h}_{4}}+%
\frac{(\underline{h}_{4}^{\circ })^{2}}{2(\underline{h}_{4})^{2}}+\left(
\frac{\underline{h}_{4}^{\circ }}{2\underline{h}_{4}}\right) ^{\circ }-\frac{%
\underline{h}_{4}^{\circ }}{2\underline{h}_{4}}\frac{\underline{h}%
_{4}^{\circ }}{2\underline{h}_{4}}+(\frac{\underline{h}_{4}^{\circ }}{2%
\underline{h}_{4}})^{2}\right) +\frac{\mathbf{\partial }_{k}\underline{h}%
_{4}^{\circ }}{2\underline{h}_{4}}-\frac{\underline{h}_{4}^{\circ }\mathbf{%
\partial }_{k}\underline{h}_{4}}{2(\underline{h}_{4})^{2}} \\
&&+\frac{\mathbf{\partial }_{k}h_{4}}{2h_{4}}\frac{\underline{h}_{4}^{\circ }%
}{2\underline{h}_{4}}+\frac{\mathbf{\partial }_{k}\underline{h}_{4}}{2%
\underline{h}_{4}}\frac{\underline{h}_{4}^{\circ }}{2\underline{h}_{4}}-%
\frac{\partial _{k}\underline{h}_{4}^{\circ }}{2\underline{h}_{4}}+\frac{%
\underline{h}_{4}^{\circ }\partial _{k}\underline{h}_{4}}{2(\underline{h}%
_{4})^{2}}-\frac{\mathbf{\partial }_{k}h_{4}}{2h_{4}}\frac{\underline{h}%
_{4}^{\circ }}{2\underline{h}_{4}}-\frac{\underline{h}_{4}^{\circ }}{2%
\underline{h}_{4}}\frac{\mathbf{\partial }_{k}\underline{h}_{4}}{2\underline{%
h}_{4}},
\end{eqnarray*}%
} and prove equations (\ref{eq4}).

For the coefficients
\begin{equation*}
\widehat{R}_{\ jka}^{i}=\frac{\partial \widehat{L}_{jk}^{i}}{\partial y^{k}}%
-(\frac{\partial \widehat{C}_{ja}^{i}}{\partial x^{k}}+\widehat{L}_{lk}^{i}%
\widehat{C}_{ja}^{l}-\widehat{L}_{jk}^{l}\widehat{C}_{la}^{i}-\widehat{L}%
_{ak}^{c}\widehat{C}_{jc}^{i})+\widehat{C}_{jb}^{i}\widehat{T}_{ka}^{b}
\end{equation*}
from (\ref{dcurv}), we obtain zero values because $\widehat{C}_{jb}^{i}=0$
and $\widehat{L}_{jk}^{i}$ do not depend on $y^{k}.$ We obtain $\widehat{R}%
_{ja}=\widehat{R}_{\ jia}^{i}=0$.

Taking $\widehat{R}_{\ bcd}^{a}$ from (\ref{dcurv}) and contracting the
indices in order to compute the Ricci coefficients,
\begin{equation*}
\widehat{R}_{bc}=\frac{\partial \widehat{C}_{bc}^{d}}{\partial y^{d}}-\frac{%
\partial \widehat{C}_{bd}^{d}}{\partial y^{c}}+\widehat{C}_{bc}^{e}\widehat{C%
}_{e}-\widehat{C}_{bd}^{e}\widehat{C}_{ec}^{d}.
\end{equation*}
We have
\begin{equation*}
\widehat{R}_{bc}=(\widehat{C}_{bc}^{3})^{\ast }+(\widehat{C}%
_{bc}^{4})^{\circ }-\partial _{c}\widehat{C}_{b}+\widehat{C}_{bc}^{3}%
\widehat{C}_{3}+\widehat{C}_{bc}^{4}\widehat{C}_{4}-\widehat{C}_{b3}^{3}%
\widehat{C}_{3c}^{3}-\widehat{C}_{b4}^{3}\widehat{C}_{3c}^{4}-\widehat{C}%
_{b3}^{4}\widehat{C}_{4c}^{3}-\widehat{C}_{b4}^{4}\widehat{C}_{4c}^{4}.
\end{equation*}
There are nontrivial values,{\small
\begin{eqnarray*}
\widehat{R}_{33} &=&\left( \widehat{C}_{33}^{3}\right) ^{\ast }+\left(
\widehat{C}_{33}^{4}\right) ^{\circ }-\widehat{C}_{3}^{\ast }+\widehat{C}%
_{33}^{3}\widehat{C}_{3}+\widehat{C}_{33}^{4}\widehat{C}_{4}-\widehat{C}%
_{33}^{3}\widehat{C}_{33}^{3}-2\widehat{C}_{34}^{3}\widehat{C}_{33}^{4}-%
\widehat{C}_{34}^{4}\widehat{C}_{43}^{4} \\
&=&-\frac{1}{2}\frac{h_{4}^{\ast \ast }}{h_{4}}+\frac{1}{4}\frac{%
(h_{4}^{\ast })^{2}}{(h_{4})^{2}}+\frac{1}{4}\frac{h_{3}^{\ast }}{h_{3}}%
\frac{h_{4}^{\ast }}{h_{4}}, \\
\widehat{R}_{44} &=&\left( \widehat{C}_{44}^{3}\right) ^{\ast }+\left(
\widehat{C}_{44}^{4}\right) ^{\circ }-\partial _{4}\widehat{C}_{4}+\widehat{C%
}_{44}^{3}\widehat{C}_{3}+\widehat{C}_{44}^{4}\widehat{C}_{4}-\widehat{C}%
_{43}^{3}\widehat{C}_{34}^{3}-2\widehat{C}_{44}^{3}\widehat{C}_{34}^{4}-%
\widehat{C}_{44}^{4}\widehat{C}_{44}^{4} \\
&=&-\frac{1}{2}\frac{h_{4}^{\ast \ast }}{h_{3}}\underline{h}_{4}+\frac{1}{4}%
\frac{h_{3}^{\ast }h_{4}^{\ast }}{(h_{3})^{2}}\underline{h}_{4}+\frac{1}{4}%
\frac{h_{4}^{\ast }}{h_{3}}\frac{h_{4}^{\ast }}{h_{4}}\underline{h}_{4}.
\end{eqnarray*}%
}These formulas are equivalent to nontrivial v--coefficients of the Ricci
d--tensor,%
\begin{eqnarray*}
\widehat{R}_{~3}^{3} &=&\frac{1}{h_{3}\underline{h}_{3}}\widehat{R}_{33}=%
\frac{1}{2h_{3}h_{4}}[-h_{4}^{\ast \ast }+\frac{(h_{4}^{\ast })^{2}}{2h_{4}}+%
\frac{h_{3}^{\ast }h_{4}^{\ast }}{2h_{3}}]\frac{1}{\underline{h}_{3}}, \\
\widehat{R}_{~4}^{4} &=&\frac{1}{h_{4}\underline{h}_{4}}\widehat{R}_{44}=%
\frac{1}{2h_{3}h_{4}}[-h_{4}^{\ast \ast }+\frac{(h_{4}^{\ast })^{2}}{2h_{4}}+%
\frac{h_{3}^{\ast }h_{4}^{\ast }}{2h_{3}}]\frac{1}{\underline{h}_{3}},
\end{eqnarray*}%
i.e. to the equations (\ref{eq2}).

\subsection{Geometric data for diagonal MGYMH configurations}

\label{aseymheq}The diagonal ansatz for generating solutions of the system (%
\ref{ym1})--(\ref{heq3}) is fixed in the form%
\begin{eqnarray}
\ ^{\circ }\mathbf{g} &=&\ ^{\circ }g_{i}(x^{1})dx^{i}\otimes dx^{i}+\
^{\circ }h_{a}(x^{1},x^{2})dy^{a}\otimes dy^{a}=  \label{ansatz1} \\
&=&q^{-1}(r)dr\otimes dr+r^{2}d\theta \otimes d\theta +r^{2}\sin ^{2}\theta
d\varphi \otimes d\varphi -\sigma ^{2}(r)q(r)dt\otimes dt,  \nonumber
\end{eqnarray}%
where the coordinates and metric coefficients are parameterized,
respectively,
\begin{eqnarray*}
u^{\alpha } &=&(x^{1}=r,x^{2}=\theta ,y^{3}=\varphi ,y^{4}=t), \\
\ ^{\circ }g_{1} &=&q^{-1}(r),\ ^{\circ }g_{2}=r^{2},\ ^{\circ
}h_{3}=r^{2}\sin ^{2}\theta ,\ ^{\circ }h_{4}=-\sigma ^{2}(r)q(r)
\end{eqnarray*}%
for $q(r)=1-$ $2m(r)/r-\overline{\Lambda }r^{2}/3,$ where $\overline{\Lambda
}$ is a cosmological constant. The function $m(r)$ is interpreted as the
total mass--energy within the radius $r$ which for $m(r)=0$ defines an empty
de Sitter, $dS,$ space written in a static coordinate system with a
cosmological horizon at $r=r_{c}=\sqrt{\frac{3}{\overline{\Lambda }}}.$ The solution
of Yang--Mills equations (\ref{ym1}) associated to the quadratic metric
element (\ref{ansatz1}) is defined by a single magnetic potential $\omega
(r),$
\begin{equation}
\ ^{\circ }A=\ ^{\circ }A_{2}dx^{2}+\ ^{\circ }A_{3}dy^{3}=\frac{1}{2e}\left[
\omega (r)\tau _{1}d\theta +(\cos \theta \ \tau _{3}+\omega (r)\tau _{2}\sin
\theta )\ d\varphi \right] ,  \label{ans1a}
\end{equation}%
where $\tau _{1},\tau _{2},\tau _{3}$ are Pauli matrices. The corresponding
solution of (\ref{heq3}) is given by
\begin{equation}
\Phi =\ ^{\circ }\Phi =\varpi (r)\tau _{3}.  \label{ans1b}
\end{equation}%
Explicit values for the functions $\sigma (r),q(r),\omega (r),\varpi (r)$
have been found, for instance, in Ref. \cite{br3} following certain
considerations that the data (\ref{ansatz1}), (\ref{ans1a}) and (\ref{ans1b}%
), i.e. $\left[ \ ^{\circ }\mathbf{g}(r)\mathbf{,}\ ^{\circ }A(r),\ \
^{\circ }\Phi (r)\right] ,$ define physical solutions with diagonal metrics
depending only on radial coordinate. \ A well known diagonal
Schwarz\-schild--de Sitter solution of (\ref{ym1})--(\ref{heq3}) is that
given by data%
\begin{equation*}
\omega (r)=\pm 1,\sigma (r)=1,\phi (r)=0, \varkappa(r)=1-2M/r-\overline{%
\Lambda }r^{2}/3
\end{equation*}%
which defines a black hole configuration inside a cosmological horizon
because $q(r)=0$ has two positive solutions and $M<1/3\sqrt{\overline{%
\Lambda }}.$

\end{document}